\documentclass{IEEEtran}
\usepackage{multicol}

\usepackage{amsmath,epsfig,comment}
\usepackage{amsthm}
\usepackage{tcolorbox}
\usepackage{amssymb}
\usepackage{multirow}
\usepackage{mathrsfs}
\usepackage{amsmath}
\usepackage{cases} 
\usepackage{amssymb,amsmath,cite}
\usepackage{epsfig}
\usepackage{color}
\usepackage{bm}
\usepackage{graphicx,subfigure}
\usepackage{algorithm}
\usepackage{algorithmic}
\usepackage{multirow}
\usepackage{soul}

\usepackage{extarrows}

\def\bb{\mathbf{b}}
\def\bC{\mathbf{C}}

\def\blam{\boldsymbol{\lambda}}
\def\I{\mathcal{I}}
\def\RR{\mathcal{R}}
\def\bA{\mathbf{A}}
\def\bY{\mathbf{Y}}

\def\bX{\mathbf{X}}
\def\bD{\mathbf{D}}
\def\i{\mathbf{i}}
\def\bU{\mathbf{U}}

\def\I{\mathcal{I}}
\def\R{\mathbb{R}}

\def\C{\mathbb{C}}
\def\H{\mathbf{H}}

\def\Re{\mathcal{R}}
\def\Im{\mathcal{I}}

\def\x{\mathbf{x}}

\def\h{\mathbf{h}}

\def\bA{\mathbf{A}}
\def\bM{\mathbf{M}}
\def\bZ{\mathbf{Z}}
\def\bD{\mathbf{D}}

\def\bH{\mathbf{H}}
\def\bB{\mathbf{B}}

\def\bh{\mathbf{h}}
\def\bthe{\boldsymbol{\Theta}}

\def\bw{\mathbf{w}}

\def\u{\mathbf{u}}
\def\bv{\mathbf{v}}
\def\bw{\mathbf{w}}

\newtheorem{lemma}{Lemma}
\newtheorem{remark}{Remark}

\newtheorem{proposition}{Proposition}

\newtheoremstyle{noparens}%
  {}{}%
  {\itshape}{}%
  {\bfseries}{.}%
  { }%
  {\thmname{#1}\thmnumber{ #2}\mdseries\thmnote{ #3}}

\theoremstyle{noparens}

\title{Beyond-Diagonal RIS Architecture Design and Optimization under Physics-Consistent Models}
\author{\IEEEauthorblockN{Zheyu Wu, Matteo Nerini, \IEEEmembership{Member, IEEE}, and  Bruno Clerckx, \IEEEmembership{Fellow, IEEE}}
\thanks{Received 14 October 2025; revised 3 February 2026; accepted 12 March 2026. This work has been supported in part  by the UK Research and Innovation (UKRI) under Grant EP/Y004086/1, Grant EP/X040569/1, Grant EP/Y037197/1, Grant EP/X04047X/1, and Grant EP/Y037243/1. \emph{(Corresponding author: Bruno Clerckx.)}}

    	\thanks{Z. Wu and M. Nerini are with the Department of Electrical and Electronic Engineering, Imperial College London, London, SW7 2AZ, U.K. (email: \{zheyu.wu, m.nerini20\}@imperial.ac.uk). }
	
	\thanks{ B. Clerckx is with the Department of Electrical and Electronic Engineering, Imperial College London, London, SW7 2AZ, U.K, and also with Department of Electronic Engineering, Kyung Hee University, Yongin-si, Gyeonggi-do 17104, South Korea. (email: b.clerckx@imperial.ac.uk) }
  }
  
\date{today}
\begin{document}
\maketitle
\begin{abstract}

Reconfigurable intelligent surface (RIS) is a promising technology for  future wireless communication systems. Conventional RIS is constrained to a diagonal scattering matrix, which limits its flexibility. Recently, beyond-diagonal RIS (BD-RIS) has been proposed as a more general RIS architecture class that allows inter-element connections and shows great potential for performance improvement. Despite extensive progress on BD-RIS, most existing studies rely on simplified channel models that ignore practical electromagnetic (EM) effects such as mutual coupling and impedance mismatching. To address this gap, this paper investigates the architecture design and optimization of BD-RIS under the general physics-consistent model derived with multiport network theory in recent literature. Building on a compact reformulation of this model, we show that band-connected RIS achieves the same channel-shaping capability as fully-connected RIS, which extends existing results obtained  for conventional channel models. We then develop optimization methods under the general  physics-consistent model; specifically, we derive closed-form solutions for single-input single-output (SISO) systems, propose a globally optimal semidefinite relaxation (SDR)–based algorithm for single-stream multi-input multi-output (MIMO) systems, and design an efficient alternating direction method of multipliers (ADMM)–based algorithm for multiuser MIMO systems. Using the proposed algorithms, we conduct comprehensive simulations to evaluate the impact of various EM effects and approximations. The results indicate that the commonly adopted unilateral approximation provides sufficient accuracy in RIS-aided systems and can therefore be readily adopted to simplify the channel model, whereas mutual coupling among RIS elements should be properly taken into account in channel modeling. 

\end{abstract}
\begin{IEEEkeywords}
Beyond-diagonal reconfigurable intelligent surface, multiport network theory, physics-consistent model,  RIS architecture, optimization.
\end{IEEEkeywords}
\section{Introduction}
Reconfigurable intelligent surface (RIS) is a planar surface composed of a large number of low-cost, nearly passive reflecting elements. By coordinately reconfiguring these elements, the RIS can manipulate the propagation of incident electromagnetic waves, thereby shaping the wireless channel to enhance signal quality and coverage. 
It is widely recognized as a key technology to meet the growing demand for high-capacity, reliable, and energy-efficient communication in the 6G era \cite{RIS_tutorial,RIS_survey,RIS_survey2}.

RIS can be modeled as an antenna array connected to a reconfigurable impedance network \cite{BDRIS}, and its scattering behavior is typically characterized by the scattering matrix of the network. The conventional RIS adopts a single-connected architecture, where each reconfigurable element is connected to a tunable impedance to ground without inter-element connections. Consequently, its scattering matrix, also known as the phase shift matrix, is diagonal. To enhance the wave manipulation capability of conventional RIS, an advanced generalization named beyond-diagonal RIS (BD-RIS) has been recently proposed. 
By introducing interconnections among reconfigurable elements with tunable impedances, BD-RIS provides significantly greater design flexibility than conventional RIS and is mathematically characterized by a scattering matrix beyond the diagonal form (hence the term BD-RIS).

The concept of BD-RIS was first proposed in \cite{BDRIS,group_conn}, where the reconfigurable elements are interconnected through tunable impedances in groups (referred to as group-connected RIS) or all together (referred to as fully-connected RIS). In particular, the conventional RIS can be viewed as a special case of group-connected RIS in which each element forms an individual group. Thanks to the additional design flexibility provided by inter-element connections, BD-RIS  is able  to achieve higher channel gain \cite{BDRIS} and sum-rate \cite{group_conn} compared to conventional RIS. In addition, unlike conventional RIS which only reflects signals, BD-RIS can work under multiple modes \cite{group_conn}, including transmissive, reflective, and hybrid modes, and can also be deployed in a multi-sector setup to greatly enhance signal coverage \cite{coverage}. Since then, BD-RIS has attracted significant research attention. Extensive studies have been devoted to designing efficient optimization algorithms \cite{closeform,PDD,twostage,wu,MIMObcapacity},  proposing new BD-RIS architectures that balance complexity and performance \cite{tree,qstem,graph}, assessing the performance of different BD-RIS architectures under non-ideal hardware \cite{discrete,mutualcoupling,mutualcoupling2,lossy,lossy_peng}, and exploring new application scenarios aided by BD-RIS \cite{RSMA0,UAV,attack,Zheng2025ISAC};  see \cite{BDRIS_survey,tutorialbdris} for more detailed summaries.
 
Despite substantial research progress on BD-RIS in various aspects, most studies rely on simplified channel models that do not fully account for the electromagnetic (EM) properties of the system.  
To address this limitation, some works have developed physics-consistent models for RIS-aided\footnote{The terminology ``RIS'' includes both conventional RIS and BD-RIS.} communication systems that capture practical non-idealities such as imperfect matching and mutual coupling among the transmit, RIS, and receive antennas \cite{BDRIS,Zparameter,Nossek,Abrardo,generalmodel}.  In particular, the authors in \cite{BDRIS} and \cite{Zparameter} developed RIS-aided channel models based on multiport network theory with the scattering and impedance parameters, respectively. The equivalence between these two parameters has been later analyzed in \cite{Nossek,Abrardo,generalmodel}, where in \cite{generalmodel} a universal framework has been developed showing also the equivalence with the admittance parameters.
However,  the general physics-consistent models based on multiport network theory, namely accounting for all the practical EM non-idealities,  yield complicated expressions in which the role of the RIS is not explicitly visible (see Section \ref{sec:generalmodel}). To make these models more tractable, various approximation techniques have been employed, such as adopting the unilateral approximation \cite{UA} (which assumes the transmission distance between a transmitter and a receiver to be large enough such that the currents at the transmitter are not affected by the currents at the receiver), neglecting mutual coupling, and assuming perfectly matched antenna arrays.  In particular, the conventional channel model adopted in the literature has been shown to be an approximation of the general physics-consistent model that incorporates all the aforementioned simplifications, and additionally neglects the specular reflection caused by the structural scattering of the RIS \cite{Nossek,Abrardo,generalmodel}.

To date, there are still several open questions regarding the physics-consistent models of RIS-aided systems. First, since the RIS affects the channel model in a complicated manner, all existing analyses rely on the unilateral approximation \cite{UA} for simplification. The literature often assumes that the unilateral approximation remains accurate in practice, as transmission distances are believed to be large enough. However, this has still not been rigorously verified in RIS-aided systems. 
  Second, the works \cite{mutualcoupling,mutualcoupling2} have demonstrated that mutual coupling among RIS elements has a noticeable impact on the channel gain of RIS-aided SISO systems (under the unilateral approximation). However, the study of the fundamental limits of RIS with mutual coupling in the more general multi-antenna and multiuser systems remains a research gap.
Third, the optimality of the RIS architectures proposed in \cite{tree} for multi-input single-output (MISO) systems and in \cite{graph} for general multi-input multi-output (MIMO) systems holds only for the conventional channel model. It remains unclear whether these architectures are still optimal when accounting for practical non-idealities such as  imperfect matching and mutual coupling captured by  physics-consistent models.

In this paper, we address the aforementioned research gaps by working with the general physics-consistent model, namely the one presented in Section \ref{sec:generalmodel}. We identify the optimal RIS architectures  and develop efficient optimization algorithms under this general model. Our results and algorithms unify and extend existing studies based on conventional simplified channel models. Using the proposed algorithms, we further evaluate the impact of various EM effects and commonly adopted approximations on system performance. The main contributions of this paper are summarized as follows.  

\begin{itemize}
\item \emph{A compact form of the general physics-consistent model}. 
First, by carefully examining the structure  of the general physics-consistent model in \cite{generalmodel}, we transform it into a compact form as follows: $\mathbf{H}=\bar{\mathbf{H}}_{RT}+\bar{\mathbf{H}}_{RI}\bar{\bthe}\bar{\mathbf{H}}_{IT}$. This  form shares the same mathematical structure as the  conventional channel model commonly adopted in the literature, 
but the expression of each term is more complicated: 
they are functions of the impedance matrices of/among the transmitter, RIS, and receiver. The effect of RIS is fully captured by $\bar{\bthe}$, but in a more involved manner (it does not correspond to the actual scattering matrix of the RIS); see Proposition \ref{pro:general} for a rigorous statement. This compact form makes the effect of the RIS explicitly visible and facilitates further analysis and optimization.

\item \emph{Architecture design under the general physics-consistent model.} Building on the compact form, we identify the optimal RIS architecture under the general physics-consistent model. We show that the band-connected RIS, which was originally proposed in \cite{graph} as the optimal RIS architecture under the conventional channel model, achieves the same performance as fully-connected RIS for MIMO systems under the general physics-consistent model as well. In particular, tree-connected RIS is optimal for MISO systems. These results  generalize existing results in \cite{tree,graph,mutualcoupling2} from conventional channel models to the general physics-consistent model and indicate that practical non-idealities such as mutual coupling and imperfect matching do not affect the optimality of  RIS architectures.

\item \emph{Optimization under the general physics-consistent model.} We study the optimization of RIS under the general physics-consistent  model for different systems. For SISO systems, we show that the approach in \cite{mutualcoupling2} can be directly extended to obtain closed-form optimal solutions that maximize the channel gain of the general physics-consistent model.  For single-stream MIMO systems, we carefully reformulate the receive power maximization problem and develop a novel semidefinite relaxation (SDR)-based algorithm. We prove that the proposed method attains the global optimum. To the best of our knowledge, this is the first approach that guarantees global optimality, whereas existing methods fail to do so even under the conventional channel model.
For  multiuser MIMO systems, we extend the algorithm in \cite{wu} for sum-rate maximization to the  general physics-consistent model. 

\item \emph{Performance assessment of RIS under the general physics-consistent model.} 
Using our optimization framework, we study the impact of various EM effects and approximations through simulations. Our results reveal that stronger mutual coupling can enhance system performance in SISO, MIMO, and multiuser MIMO systems under Rayleigh fading channels, thereby extending the conclusion of \cite{mutualcoupling2} which considered only the SISO case.  We also show that neglecting the mutual coupling among RIS elements in optimization leads to performance degradation. In contrast, the unilateral approximation remains highly accurate as long as the transmission distance is on the order of a few wavelengths. Hence, in typical RIS-aided wireless communication scenarios, this approximation can be reliably employed to simplify the channel model.
\end{itemize}

\emph{Organization}: The rest of the paper is organized as follows.  In Section \ref{sec:model}, we introduce  physics-consistent  models of RIS-aided communication systems following previous literature. In Section \ref{sec:3}, we propose a compact reformulation of the general physics-consistent model. Sections \ref{sec:4} and \ref{sec:opt} discuss architecture design and optimization under the general  physics-consistent model, respectively. In Section \ref{sec:simulation}, we present simulation results to demonstrate  the impact of various EM effects and approximations. Finally, the paper is concluded in Section \ref{sec:7}.

\emph{Notations}: Throughout the paper, we use $x$, $\x$, $\mathbf{X}$, and $\mathcal{X}$ to denote a scalar, column vector, matrix, and set, respectively. For a vector $\x$, $[\x]_i$ denotes its $i$-th entry. For a matrix $\mathbf{X}$,  $[\bX]_{\mathcal{S}_1,\mathcal{S}_2}$ denotes its submatrix with rows indexed by $\mathcal{S}_1$ and columns indexed by $\mathcal{S}_2$, and in particular, $[\mathbf{X}]_{i,j}$ denotes its $(i,j)$-th entry;  $\bX^T$, $\bX^H$, $\bX^{-1}$, $\text{tr}(\bX)$, $\text{vec}(\bX)$, $\RR(\bX)$, $\I(\bX)$, and $\|\bX\|_F$  return the transpose,  the Hermitian transpose, the inverse, the trace, the column-wise vectorization, the real part, the imaginary part, and the Frobenius norm of $\bX$,  respectively; $\bX\succeq (\succ)\mathbf{0}$ means that $\bX$ is positive semidefinite (definite); $\bX^{-\frac{1}{2}}$ refers to the square root of $\bX$ (given that $\bX$ is positive definite).     The notation  $\|\cdot\|_2$  denotes the  $\ell_2$-norm of a vector or the spectral norm of a matrix. The symbols $\mathbf{0}$ and  $\mathbf{I}$  refer to an all-zero matrix and an identity matrix, respectively. Finally, $\mathsf{i}$ represents the imaginary unit.

\section{Physics-Consistent RIS-Aided Channel Models Based on Multiport Network Analysis}\label{sec:model}
Consider a MIMO\footnote{The term ``MIMO" is used here in a general sense, including both point-to-point MIMO and multiuser MIMO systems.} wireless communication system with $N_T$ transmit antennas and $N_R$ receive antennas aided by a  RIS consisting of $N_I$ elements. As discussed in \cite{generalmodel}, the overall wireless channel can be regarded as an $N$-port network, where $N=N_T+N_I+N_R$, and can be modeled based on the multiport network theory \cite{microwavebook}. In this section, we introduce the physics-consistent channel models for RIS-aided systems based on the impedance parameters developed in \cite{generalmodel,Nossek,Zparameter},  along with several widely adopted approximations and the corresponding simplified models. 
\begin{figure}
\includegraphics[width=0.45\textwidth]{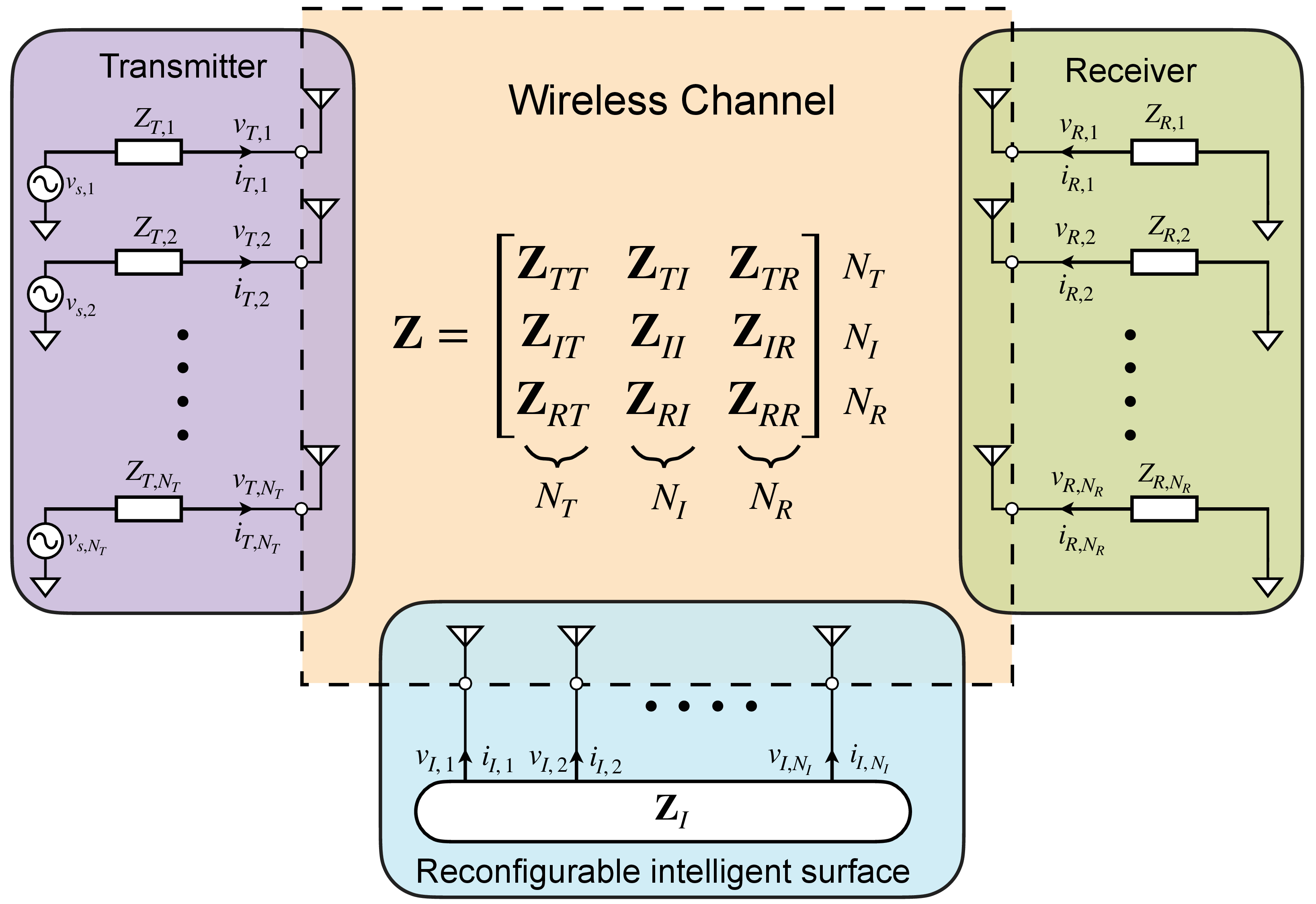}
\caption{Multiport model of a RIS-aided MIMO system.}
\centering
\label{multiport}
\end{figure}
\subsection{The General Physics-Consistent RIS-Aided Channel Model}\label{sec:generalmodel}
Consider the $N$-port network model for RIS-aided MIMO systems depicted in Fig. \ref{multiport}.  Let $\bv_{i}\in\C^{N_i}$,  $i\in\{T,I,R\}$, denote the voltages at the transmit, RIS, and receive antenna arrays, respectively. The overall channel matrix $\bH \in \mathbb{C}^{N_R \times N_T}$ is defined as the matrix that characterizes the linear relationship between the transmit and receive voltages, i.e.,
\begin{equation}\label{vR=HvT}
\bv_R=\bH\bv_T.
\end{equation} 
In the following, we derive the expression for $\bH$, 
briefly recalling the derivations of \cite{generalmodel}.

According to multiport network theory \cite{microwavebook}, the electrical behavior of a multiport network can be characterized by the impedance matrix of the network $\bZ\in\C^{N\times N}$, which relates the voltage at the $N$ ports $\bv\in\C^N$ to the current at the $N$ ports $\mathbf{i}\in\C^N$ through \begin{equation}\label{v=Zi}
\mathbf{v}=\mathbf{Z}\mathbf{i}.
\end{equation}
For the RIS-aided MIMO system considered in Fig. \ref{multiport},  
$\bv=[\bv_T^T,\bv_I^T,\bv_R^T]^T$, $\i=[\i_T^T,\i_I^T,\i_R^T]^T$,  where $\bv_{i}\in\C^{N_i}$ and $\i_i\in\C^{N_i}$, $i\in\{T,I,R\}$, denote the voltage and current at the transmit, RIS, and receive antenna arrays, respectively,  
and $\bZ$ is given by 
$$\bZ=\left[\begin{matrix}\bZ_{TT}&\bZ_{TI}&\bZ_{TR}\\\bZ_{IT}&\bZ_{II}&\bZ_{IR}\\\bZ_{RT}&\bZ_{RI}&\bZ_{RR}\end{matrix}\right],$$ 
where  $\bZ_{TT}\in\C^{N_T\times N_T}$, $\bZ_{II}\in\C^{N_I\times N_I}$, and $\bZ_{RR}\in\C^{N_R\times N_R}$ are the impedance matrices of the antenna arrays at transmitter, RIS, and receiver, respectively.  Their diagonal entries represent antenna self-impedance, while the off-diagonal entries capture antenna mutual coupling. 
  The submatrices  $\bZ_{IT}\in\C^{N_I\times N_T}$, $\bZ_{RI}\in\C^{N_R\times N_I}$, and  $\bZ_{RT}\in\C^{N_R\times N_T}$ refer to the impedance matrices from transmitter to RIS, from RIS to receiver, and from transmitter  to receiver, respectively. Similarly,  $\bZ_{TI}\in\C^{N_T\times N_I}$,  $\bZ_{IR}\in\C^{N_I\times N_R}$, and $\bZ_{TR}\in\C^{N_T\times N_R}$  refer to the  impedance matrices from RIS to transmitter, from receiver to RIS, and from receiver to transmitter, respectively.  Assuming channel reciprocity, we have  $\bZ_{RT}=\bZ_{TR}^T$, $\bZ_{RI}=\bZ_{IR}^T$, $\bZ_{IT}=\bZ_{TI}^T$. 
  
  By individually analyzing the circuits at the transmitter, RIS, and receiver, we can further characterize the relationship between $\bv_i$ and $\i_i$, where $i \in \{T, I, R\}$.  First, at the transmitter, each antenna is connected with a source voltage and a source impedance, yielding the following relationship:
\begin{equation}\label{vT}
\bv_T=\bv_{s,T}-\bZ_T\i_T,
\end{equation}
where $\bv_{s,T}=[v_{s,1},v_{s,2},\dots,v_{s,N_T}]\in\C^{N_T}$ and $\bZ_T=\text{diag}(Z_{T,1},Z_{T,2}\dots,Z_{T,N_T})\in\C^{N_T\times N_T}$; $v_{s,n}$ and $Z_{T,n}$ are the source voltage and source impedance of the $n$-th transmit antenna, respectively; see Fig. \ref{multiport}.

The RIS is modeled as $N_I$ antennas connected to an $N_I$-port reconfigurable impedance network \cite{generalmodel}. As a result, $\bv_I$ and $\mathbf{i}_I$ are related by 
\begin{equation}\label{vI}
\bv_I=-\mathbf{Z}_I\mathbf{i}_I,
\end{equation}
where $\mathbf{Z}_I\in\C^{N_I\times N_I}$ is the impedance matrix of the reconfigurable impedance network. In this paper, we adopt the common assumption that the reconfigurable impedance network is reciprocal and lossless, which implies that $\bZ_I$ is symmetric and purely imaginary \cite[Chapter 4]{microwavebook}. 

At the receiver, each antenna is connected with a load  impedance, and $\bv_R$ and $\mathbf{i}_R$ satisfies 
\begin{equation}\label{vR}
\bv_R=-\mathbf{Z}_R\mathbf{i}_R,
\end{equation}
where $\bZ_R=\text{diag}(Z_{R,1},Z_{R,2}\dots,Z_{R,N_R})\in\C^{N_R\times N_R}$;  $Z_{R,n}$ is the load impedance of the  $n$-th receive antenna. 

Combining \eqref{vR=HvT} -- \eqref{vR}, it has been shown in \cite[Section III]{generalmodel}  that 
\begin{equation}\label{general_model}
\mathbf{H}=\widetilde{\mathbf{Z}}_{R T} \widetilde{\mathbf{Z}}_{T T}^{-1}.
\end{equation}
Here, $\widetilde{\mathbf{Z}}:=\left(\mathbf{I}+{\mathbf{Z}}_0 \mathbf{Z}^{-1}\right)^{-1}$ and is partitioned as
$$
\widetilde{\mathbf{Z}}=\left[\begin{array}{ccc}
\widetilde{\mathbf{Z}}_{T T} & \widetilde{\mathbf{Z}}_{T I} & \widetilde{\mathbf{Z}}_{T R} \\
\widetilde{\mathbf{Z}}_{I T} & \widetilde{\mathbf{Z}}_{I I} & \widetilde{\mathbf{Z}}_{I R} \\
\widetilde{\mathbf{Z}}_{R T} & \widetilde{\mathbf{Z}}_{R I} & \widetilde{\mathbf{Z}}_{R R}
\end{array}\right],
$$
where $\bZ_0=\operatorname{blkdiag}(\bZ_T,\bZ_I,\bZ_R)$.

The above derivation makes no assumption on the communication scenario or the transmit, RIS, and receive antennas.  Thus, the model in \eqref{general_model} is accurate.  However, its analytical expression is very complex.  To improve tractability, several simplified models have been proposed  based on different assumptions. In the following, we introduce  commonly adopted assumptions and the corresponding approximate models in the literature. 

\subsection{Approximations and Approximate RIS-aided Channel Models}\label{sec:approximation}
 \emph{(A1) Unilateral approximation}:  The ``unilateral approximation'' assumes that the distances between the transmitter, RIS, and receiver are sufficiently large such that the electromagnetic effect from the receiving devices to the transmitting devices are negligible \cite{UA}. Under this assumption, we can set $\bZ_{TI}= \mathbf{0}, \bZ_{TR}=\mathbf{0}, \bZ_{IR}=\mathbf{0}$ since these feedback links do not affect the channel matrix, which leads to the following approximate channel model \cite[Eq. (64)]{generalmodel}:
$$\bH_{\text{app,1}}\hspace{-0.06cm}=\hspace{-0.06cm}\bZ_R(\bZ_R+\bZ_{RR})^{-1}\hspace{-0.1cm}\left(\bZ_{RT}\hspace{-0.05cm}-\hspace{-0.1cm}\bZ_{RI}(\bZ_I\hspace{-0.06cm}+\hspace{-0.06cm}\bZ_{II})^{-1}\bZ_{IT}\right)\bZ_{TT}^{-1}.$$

 \emph{(A2) Assumptions on the transmitter and receiver}:
The following two assumptions are commonly imposed at the transmitter and receiver for further simplification.  First, we assume that the source impedances $\{Z_{T,n}\}_{1\leq n\leq N_T}$ at the transmitter and the load impedances  $\{Z_{R,n}\}_{1\leq n\leq N_R}$ at the receiver  are equal to the reference impedance $Z_0$, commonly set as $Z_0=50\;\Omega$, which gives  $\bZ_R=Z_0 \mathbf{I}$ and  $\bZ_T=Z_0\mathbf{I}$. Second, we assume that the 
 transmit and receive antennas are perfectly matched with no mutual coupling, and thus 
$\bZ_{TT}=Z_0\mathbf{I}$ and $\bZ_{RR}=Z_0\mathbf{I}.$ Under (A1) and (A2), the channel model simplifies to \cite[Eq. (76)]{generalmodel}
$$\bH_{\text{app,2}}=\frac{1}{2Z_0}\left(\bZ_{RT}-\bZ_{RI}(\bZ_I+\bZ_{II})^{-1}\bZ_{IT}\right).$$

\emph{(A3) Perfect matching and no mutual coupling among RIS}: This assumption gives $\bZ_{II}=Z_0\mathbf{I}$. Under (A1) -- (A3), the channel model simplifies to \cite[Eq. (77)]{generalmodel}
$$\bH_{\text{app,3}}=\frac{1}{2Z_0}\left(\bZ_{RT}-\bZ_{RI}(\bZ_I+Z_0\mathbf{I})^{-1}\bZ_{IT}\right).$$  

We remark that, in addition to the impedance matrix $\mathbf{Z}$ ($Z$-parameter), the electrical properties of a multiport network can also be characterized by the admittance matrix $\bY$ ($Y$-parameter) and the scattering matrix $\mathbf{S}$ ($S$-parameter) \cite{microwavebook}, which are related to $\bZ$ by
\begin{equation}\label{Y=invZ}
\bY=\bZ^{-1}
\end{equation} and 
$$\mathbf{S}=(\bZ+Z_0\mathbf{I})^{-1}(\bZ-Z_0\mathbf{I}),$$
respectively. 
All the channel models introduced above (i.e., $\bH$, $\bH_{\text{app},1}$, $\bH_{\text{app},2}$, $\bH_{\text{app},3}$) can also be  expressed through impedance and scattering matrices; see details in \cite{generalmodel}.  In particular, the channel model $\bH_{\text{app},3}$, when expressed via the scattering matrix, reduces to the following conventional form:
\begin{equation}\label{Happ3}
\bH_{\text{app},3}=\mathbf{H}_{RT}+\mathbf{H}_{RI}\bthe\mathbf{H}_{IT},
\end{equation}
where $\mathbf{H}_{RT}:=\mathbf{S}_{RT}$, $\mathbf{H}_{RI}:=\mathbf{S}_{RI}$, and $\mathbf{H}_{IT}:=\mathbf{S}_{IT}$ denote the channels (scattering matrices) from transmitter to receiver, from RIS to receiver, and from transmitter to RIS, respectively.

Most existing RIS literature employs the channel model in \eqref{Happ3}, which is derived under Assumptions (A1)\,--\,(A3). A few  works also account for mutual coupling among RIS antennas and employ the approximate channel model $\mathbf{H}_{\text{app,2}}$ \cite{mutualcoupling,mutualcoupling2}. 
 To the best of our knowledge, the general form in \eqref{general_model} has not been considered in existing works, mainly due to its analytical complexity. 
  
   In this paper, we \emph{derive a compact expression for the general physics-consistent model in \eqref{general_model} and  investigate both architecture design and optimization within this general framework.} Throughout the paper, we  assume that the RIS is lossless and reciprocal, an assumption commonly adopted for reconfigurable impedance networks. By working with \eqref{general_model}, our results and algorithms are applicable without relying on any assumptions about the wireless channel or the transmit,  RIS, and receive antennas. 
   Our aim is to \emph{(i) identify the optimal RIS architecture under the general physics-consistent model; and (ii) provide a systematic comparison of different approximations and channel models using the proposed algorithms.}

  \section{A Compact Form  of the General Physics-Consistent Model}\label{sec:3}
It is difficult to perform signal processing based on \eqref{general_model}, as the design variable related to  RIS, i.e., its impedance matrix $\bZ_I$,  
is  hidden in the channel expression inside the blocks of $\widetilde{\mathbf{Z}}$.
In this section, we derive a compact form of \eqref{general_model} to facilitate further analysis.\vspace{-0.1cm}
\subsection{An Explicit Expression for the General Physics-Consistent Model} We begin by deriving an equivalent form of \eqref{general_model} in which the effect of RIS is explicit.

Observe that the impedance matrix $\bZ$ affects the model in \eqref{general_model} only through its inverse. Therefore, it is convenient to express \eqref{general_model} based on the admittance matrix $\bY=\bZ^{-1}$. Let $\bY_T=\bZ_T^{-1}$ and $\bY_R=\bZ_R^{-1}$, which are two diagonal matrices whose diagonal elements denote the source and  load admittances, respectively, and let $\bY_I=\bZ_I^{-1}$ be the admittance matrix of the $N_I$-port reconfigurable impedance network.  
The matrix $\widetilde{\bZ}$ involved in the general model \eqref{general_model} can then be expressed as 
$\widetilde{\bZ}=\left(\bY_0+\mathbf{Y}\right)^{-1}\bY_0,
$
where $\bY_0=\operatorname{blkdiag}(\bY_T,\bY_I,\bY_R)$.
Denote  $\widetilde{\bY}:=(\bY_0+\bY)^{-1}$, partitioned as   
$$\widetilde{\bY}=\left[\begin{matrix}\widetilde{\bY}_{TT}&\widetilde{\bY}_{TI}&\widetilde{\bY}_{TR}\\\widetilde{\bY}_{IT}&\widetilde{\bY}_{II}&\widetilde{\bY}_{IR}\\\widetilde{\bY}_{RT}&\widetilde{\bY}_{RI}&\widetilde{\bY}_{RR}\end{matrix}\right].$$ 
Clearly,  $\widetilde{\bZ}_{TT}=\widetilde{\bY}_{TT}\bY_T$ and $\widetilde{\bZ}_{RT}=\widetilde{\bY}_{RT}\bY_T$, 
and thus 
$$\bH=\widetilde{\bZ}_{RT}\widetilde{\bZ}_{TT}^{-1}=\widetilde{\bY}_{RT}\widetilde{\bY}_{TT}^{-1}.$$
We now derive the explicit expressions for $\widetilde{\bY}_{TT}$ and $\widetilde{\bY}_{RT}$. Suppose that the matrix $\bY$ is partitioned in the same way as $\widetilde{\bY}$,  
then 
$$\widetilde{\bY}=\left[
\begin{matrix}
\mathbf{Y}_T+\mathbf{Y}_{TT}&\mathbf{Y}_{TI}&\mathbf{Y}_{TR}\\
\mathbf{Y}_{IT}&\mathbf{Y}_I+\mathbf{Y}_{II}&\mathbf{Y}_{IR}\\
\mathbf{Y}_{RT}&\mathbf{Y}_{RI}&\mathbf{Y}_R+\mathbf{Y}_{RR}
\end{matrix}\right]^{-1}.$$
The key step is to apply the following block matrix inversion formula to $\widetilde{\bY}$: 
\begin{equation}\label{block_inverse}
\begin{aligned}
\left[\begin{matrix}
\bA & \bB \\
\bC & \bD
\end{matrix}\right]^{-1}\hspace{-0.2cm}=\left[\begin{matrix}
\mathbf{S}_\bD^{-1} & -\mathbf{S}_\bD^{-1} \bB \bD^{-1} \\
-\bD^{-1} \bC\mathbf{S}_\bD^{-1} & \bD^{-1}\hspace{-0.05cm}+\bD^{-1} \bC\mathbf{S}_\bD^{-1} \bB \bD^{-1}
\end{matrix}\right],
\end{aligned}
\end{equation} where $\mathbf{S}_\bD=\bA-\bB \bD^{-1} \bC$ is the Schur complement of $\bD$. Specifically, define the block components as 
$$
\begin{aligned}
\bA:&=\bY_T+\mathbf{Y}_{TT}, ~~~~\bB:=[\mathbf{Y}_{TI}~~\mathbf{Y}_{TR}],~\\
\bC:&=
\left[\begin{array}{c}
\mathbf{Y}_{IT}\\
\mathbf{Y}_{RT}\\
\end{array}\right],~~\bD:=\left[\begin{array}{cc}
\bY_I+\mathbf{Y}_{II}&\mathbf{Y}_{IR}\\
\mathbf{Y}_{RI}&\mathbf{Y}_R+\mathbf{Y}_{RR}\end{array}\right].
\end{aligned}
$$
Applying \eqref{block_inverse}, we get 
$$\widetilde{\bY}_{TT}=\mathbf{S}_\bD^{-1},~\widetilde{\bY}_{RT}=-[\mathbf{0}_{N_R\times N_I}~~\mathbf{I}_{N_R}]\bD^{-1} \bC\mathbf{S}_\bD^{-1}.$$
It follows that \begin{equation*}
\begin{aligned}
\bH&=\hspace{-0.05cm}-[\mathbf{0}_{N_R\times N_I}~~\mathbf{I}_{N_R}]\,\bD^{-1}\bC\\
&=\hspace{-0.05cm}-[\mathbf{0}_{N_R\times N_I}~~\mathbf{I}_{N_R}]
\left[\begin{matrix}
\bY_I+\hspace{-0.05cm}\mathbf{Y}_{II}&\hspace{-0.05cm}\mathbf{Y}_{IR}\\
\mathbf{Y}_{RI}&\hspace{-0.2cm}\bY_R\hspace{-0.05cm}+\hspace{-0.05cm}\mathbf{Y}_{RR}\end{matrix}\right]^{-1}\hspace{-0.05cm}
\left[\begin{matrix}
\mathbf{Y}_{IT}\\
\mathbf{Y}_{RT}\\
\end{matrix}\right].
\end{aligned}
\end{equation*}
Applying the block matrix inversion formula  in \eqref{block_inverse} again yields an explicit expression for 
$\bH$, which is given in \eqref{general} on top of the next page. The routine algebraic steps are omitted for brevity.
 
\begin{figure*}
\begin{equation}\label{general}
\mathbf{H}=\left(\mathbf{Y}_R+\mathbf{Y}_{RR}\right)^{-1}\left(-\mathbf{Y}_{RT}+\mathbf{Y}_{RI}\left(\mathbf{Y}_I+\mathbf{Y}_{II}-\mathbf{Y}_{IR}(\mathbf{Y}_{R}+\mathbf{Y}_{RR})^{-1}\mathbf{Y}_{RI}\right)^{-1}\left(\mathbf{Y}_{IT}-\mathbf{Y}_{IR}(\mathbf{Y}_R+\mathbf{Y}_{RR})^{-1}\mathbf{Y}_{RT}\right)\right).
\end{equation}
---------------------------------------------------------------------------------------------------------------------------------------------------------
\end{figure*}

\begin{remark}
Under the unilateral approximation introduced in Section \ref{sec:approximation}, the  model in \eqref{general} reduces to the  approximate model $\bH_{\text{app},1}.$ Specifically, with $$\bZ_{\text{app}}=\left[\begin{matrix}\bZ_{TT}&\mathbf{0}&\mathbf{0}\\\mathbf{Z}_{IT}&\mathbf{Z}_{II}&\mathbf{0}\\\mathbf{Z}_{RT}&\mathbf{Z}_{RI}&\mathbf{Z}_{RR}\end{matrix}\right],$$
the corresponding approximate admittance matrix is given by 
$$\bY_{\text{app}}=\bZ_{\text{app}}^{-1}=\left[\begin{matrix}\bY^{\text{app}}_{TT}&\mathbf{0}&\mathbf{0}\\\mathbf{Y}_{IT}^{\text{app}}&\mathbf{Y}_{II}^{\text{app}}&\mathbf{0}\\\mathbf{Y}_{RT}^{\text{app}}&\mathbf{Y}_{RI}^{\text{app}}&\mathbf{Y}_{RR}^{\text{app}}\end{matrix}\right],$$
where 
$\bY_{ii}^{\text{app}}=\bZ_{ii}^{-1},~i\in\{T,I,R\},$
$$\mathbf{Y}_{IT}^{\text{app}}=-\bZ_{II}^{-1}\bZ_{IT}\bZ_{TT}^{-1},~\mathbf{Y}_{RI}^{\text{app}}=\bZ_{RR}^{-1}\bZ_{RI}\bZ_{II}^{-1},$$
$$\bY_{RT}^{\text{app}}=\bZ_{RR}^{-1}(-\bZ_{RT}+\bZ_{RI}\bZ_{II}^{-1}\bZ_{IT})\bZ_{TT}^{-1}.$$
Here, we use the notation ``app'' to emphasize that the admittance matrix and each of its blocks are approximate, namely obtained by inverting the approximate  impedance matrix $\bZ_{\text{app}}$  with the unilateral approximation applied. 
Substituting $\bY_{\text{app}}$ into the general model in \eqref{general}, we obtain  
\begin{equation*}\label{general_app}
\mathbf{H}_{\text{app}}\hspace{-0.05cm}=\hspace{-0.1cm}\left(\mathbf{Y}_R\hspace{-0.03cm}+\hspace{-0.03cm}\mathbf{Y}_{RR}^\text{app}\right)^{-1}\hspace{-0.1cm}\left(-\mathbf{Y}_{RT}^\text{app}\hspace{-0.03cm}+\hspace{-0.03cm}\mathbf{Y}_{RI}^\text{app}\left(\mathbf{Y}_I\hspace{-0.03cm}+\hspace{-0.03cm}\mathbf{Y}_{II}^\text{app}\right)^{-1}\mathbf{Y}_{IT}^\text{app}\right).
\end{equation*}
Further substituting the expressions of $\bY_{\text{app}}$ into the above equation yields $\bH_{\text{app,1}}$.
\end{remark}

In the above derivation, we have started from the $Z$-parameter-based channel model in \eqref{general_model} and arrived at an  expression based on $Y$-parameter in \eqref{general}, in which the effect of RIS is explicitly visible through the admittance matrix $\bY_I$. The $Y$-parameter turns out to be the most suitable representation for the general physics-consistent model, as it yields the simplest analytical expression. 
 Since no constraints are imposed on $\bY_I$ in the above derivation, the channel model given in \eqref{general} does not rely on any specific characteristics of the reconfigurable impedance network. Hence, it also holds when the RIS is lossy or non-reciprocal. In the following subsection, we further derive a compact reformulation of \eqref{general} under the assumption that the RIS is lossless and reciprocal. 

\subsection{A Compact Form of the General Physics-Consistent Model}
By the assumption that the RIS is lossless and reciprocal, $\bY_I$ is a purely imaginary symmetric matrix, which can be expressed as $\bY_I = \mathrm{i} \bB_I$, where $\bB_I \in \mathbb{R}^{N_I \times N_I}$ is the susceptance matrix satisfying $\bB_I = \bB_I^T$.

The main idea is to apply a “diagonalization” technique to the matrix 
\begin{equation}\label{def:bar_YII}
\bar{\bY}_{II}:=\bY_{II} - \bY_{IR}(\bY_{RR} + \bY_{R})^{-1}\bY_{RI}
\end{equation} involved inside a matrix inversion in  \eqref{general}. This technique is inspired by \cite{Semmler,mutualcoupling2}, where the channel model with mutual coupling among RIS elements (i.e., $\bH_{\text{app},2}$) was simplified by appropriately diagonalizing $\bZ_{II}$. However, the matrix $\bar{\bY}_{II}$ in our case  is more complicated. Through a dedicated analysis, we obtain the following lemma.
\begin{lemma}\label{lem:pd}
The matrix 
$\Re(\bar{\mathbf{Y}}_{II})$ is positive definite, where $\bar{\bY}_{II}$ is defined in \eqref{def:bar_YII}.
\end{lemma}
\begin{proof}
See Appendix \ref{app:pd}.
\end{proof}
With Lemma \ref{lem:pd}, we are now ready to present the compact reformulation of \eqref{general}.
\begin{proposition}[A Compact Form of \eqref{general}]\label{pro:general}
The model in \eqref{general} can be rewritten as 
\begin{equation}\label{general_simplified}
\mathbf{H}=\bar{\mathbf{H}}_{RT}+\bar{\mathbf{H}}_{RI}\bar{\bthe}\bar{\mathbf{H}}_{IT}.
\end{equation}
In the above model, 
 \begin{equation}\label{bartheta}
\bar{\bthe}=(Y_0\mathbf{I}+\mathrm{i}\bar{\bB}_I)^{-1}(Y_0\mathbf{I}-\mathrm{i}\bar{\bB}_I),
\end{equation}
where $Y_0=1/Z_0$ is the reference admittance, $\bar{\bB}_I$ is a linear transformation of the susceptance matrix $\bB_I$ defined as 
\begin{equation}\label{barBI}
 \bar{\bB}_I=Y_0\Re(\bar{\bY}_{II})^{-\frac{1}{2}}\left(\bB_I+\Im(\bar{\bY}_{II})\right)\Re(\bar{\bY}_{II})^{-\frac{1}{2}}, 
 \end{equation}
and $\bar{\bH}_{RT}$, $\bar{\bH}_{RI}$, and $\bar{\bH}_{IT}$ are functions of $\bY$ and $\bY_R$, whose definitions are given  below: 
$$
\begin{aligned}
\bar{\mathbf{H}}_{RT}&=-\frac{1}{2Y_0}\left(\bar{\mathbf{Y}}_{RT}-\frac{1}{2Y_0}\bar{\mathbf{Y}}_{RI}\bar{\bY}_{IT}\right),\\
\bar{\mathbf{H}}_{RI}&=-\frac{\bar{\mathbf{Y}}_{RI}}{2Y_0},~~\bar{\mathbf{H}}_{IT}=-\frac{\bar{\bY}_{IT}}{2Y_0},
\end{aligned}$$
with
$$
\begin{aligned}
\bar{\bY}_{RT}&=2Y_0\left(\mathbf{Y}_R+\mathbf{Y}_{RR}\right)^{-1}\bY_{RT},\\
\bar{\bY}_{RI}&=\sqrt{2}Y_0\left(\mathbf{Y}_R+\mathbf{Y}_{RR}\right)^{-1}\bY_{RI}\Re(\bar{\bY}_{II})^{-\frac{1}{2}},\\
\bar{\bY}_{IT}&=\sqrt{2}{Y_0}\Re(\bar{\bY}_{II})^{-\frac{1}{2}}\left(\mathbf{Y}_{IT}{\color{black}\,-\mathbf{Y}_{IR}(\mathbf{Y}_R+\mathbf{Y}_{RR})^{-1}\mathbf{Y}_{RT}}\right).
\end{aligned}$$
\end{proposition} 
\begin{proof}
The proof is similar to \cite[Eq. (37)--(41)]{mutualcoupling2}.
First, by $\bY_I=\mathrm{i}\bB_I$ and the definition of $\bar{\bB}_I$ in  \eqref{barBI}, we have
$$(\bY_I+\bar{\bY}_{II})^{-1}=Y_0\Re(\bar{\bY}_{II})^{-\frac{1}{2}}(Y_0\mathbf{I}+\mathrm{i}\bar{\bB}_I)^{-1}\Re(\bar{\bY}_{II})^{-\frac{1}{2}};$$
the matrix $\Re(\bar{\bY}_{II})^{-\frac{1}{2}}$ is uniquely defined according to Lemma \ref{lem:pd}. 
Hence,  \eqref{general} can be expressed as 
\begin{equation}\label{general_2}
\mathbf{H}=\frac{1}{2Y_0}\left(-\bar{\mathbf{Y}}_{RT}+\bar{\mathbf{Y}}_{RI}\left(Y_0\mathbf{I}+\mathrm{i}\bar{\mathbf{B}}_I\right)^{-1}\bar{\mathbf{Y}}_{IT}\right).
\end{equation}
The desired result in Proposition \ref{pro:general} can then be obtained by noting the relation  
$$-\bar{\bY}_{RI}\bar{\bY}_{IT}=-2Y_0\bar{\bY}_{RI}(Y_0\mathbf{I}+\mathrm{i}\bar{\bB}_I)^{-1}\bar{\bY}_{IT}+\bar{\bY}_{RI}\bar{\bthe}\bar{\bY}_{IT}.$$
\end{proof}
Proposition \ref{pro:general} provides a significantly simplified representation of the model in \eqref{general}. Notably, a channel model with the same mathematical structure was previously derived in \cite{Semmler} to analyze the effect of a decoupling network at the RIS with mutual coupling under the unilateral approximation. The same mathematical structure has also appeared in \cite{mutualcoupling2}, where a tight upper bound on the channel gain achievable by BD-RIS has been derived under the unilateral approximation. In contrast to \cite{Semmler,mutualcoupling2}, the model in Proposition \ref{pro:general} does not rely on the unilateral approximation and is therefore more general, and does not assume the presence of a decoupling network.

The model in Proposition \ref{pro:general} shares the same form as \eqref{Happ3}, but the expression and physical meaning of each component are completely different in the two models. In \eqref{Happ3}, $\bH_{RT}$, $\bH_{RI}$, and $\bH_{IT}$ are the corresponding submatrices of the scattering matrix $\mathbf{S}$ of the wireless channel, and $\bthe$ is the scattering matrix of the RIS. 	In \eqref{general_simplified}, $\bar{\bH}_{RT}$, $\bar{\bH}_{RI}$, and $\bar{\bH}_{IT}$ are more intricate transformation of the admittance matrix $\bY$ and the load impedance $\bY_R$, and they do not exhibit an explicit relationship with the actual scattering matrix $\mathbf{S}$.
  The effect of RIS is fully captured by $\bar{\bthe}$, which, due to the relation in \eqref{bartheta}, can be viewed as the scattering matrix of a ``virtual'' $N_I$-port network. The corresponding virtual susceptance matrix, $\bar{\bB}_I$, is a linear transformation of the actual susceptance matrix of the RIS, as captured by \eqref{barBI}.
In fact, Proposition 1 implies that all existing channel models introduced in Section \ref{sec:model} can be  represented in the form of \eqref{general_simplified},  thereby extending the results in \cite{Semmler,mutualcoupling2}, while the specific expressions of the components $\bar{\bH}_{RT}$, $\bar{\bH}_{RI}$, $\bar{\bH}_{IT}$, and $\bar{\bthe}$ vary across models.

We note that the architecture design and optimization of RIS have been extensively studied under the conventional channel model $\bH_{\text{app},3}$. Building on the unified mathematical structure shared by the general physics-consistent model and $\bH_{\text{app},3}$, existing results and algorithms can be extended to the general physics-consistent model with appropriate modifications. The key difference is that,  to account for different RIS architectures,  certain architecture-specific constraints need to be imposed on the actual susceptance matrix  $\bB_I$, while the model in \eqref{general_simplified} is expressed in terms of a transformation of $\bB_I$, namely $\bar{\bB}_I$.  In the following two sections, we  identify the optimal RIS architecture and design optimization algorithms, respectively, under the general physics-consistent channel model.

 \section{Optimal RIS Architectures under the General Physics-Consistent Model}\label{sec:4}
In this section, we study architecture design of RIS under the general physics-consistent model. We first review  existing RIS architectures in Section \ref{sec:arch}, and then discusses the optimal RIS architecture under the general physics-consistent model in Section \ref{sec:optarch}.
\subsection{Existing RIS Architectures}\label{sec:arch}
As mentioned in Section \ref{sec:generalmodel}, the RIS is modeled as $N_I$ antennas connected to an $N_I$-port reconfigurable impedance network. Different connection types among the RIS elements lead to different architectures. It is convenient to characterize the RIS architecture using the admittance matrix $\bY_I$ of the reconfigurable impedance network, as a non-zero entry $[\bY_I]_{i,j}$ indicates a connection between the $i$-th and $j$-th RIS elements, and zero entry indicates no connection \cite{tree,graph}. In the following, we introduce exiting RIS architectures.


\subsubsection{Group-Connected RIS} 
In group-connected RIS, the RIS elements are uniformly divided into $G$ groups, where elements within each group are interconnected. In this case, the admittance matrix $\bY_I$ is a block diagonal matrix, expressed as $\bY_I=\text{blkdiag}(\bY_{I,1},\dots,\bY_{I,G})$, where $\bY_{I,g}\in\C^{({N_I}/{G})\times ({N_I}/{G})}$ for all $g=1,2,\dots, G$. Two extreme cases of the group-connected RIS are (i) single-connected RIS (conventional RIS), where $G=N_I$, i.e., there is no interconnection between RIS elements; and (ii) fully-connected RIS, where $G=1$, i.e., all RIS elements are inter-connected. 

\subsubsection{Tree-Connected RIS} Tree-connected RIS refers to the architecture whose graph representation, with RIS elements as vertices and interconnections as edges, is a tree graph \cite{tree}. A representative example of tree-connected RIS is the tridiagonal RIS, where each RIS element is only connected to its neighboring element. In this case, 
\begin{equation}\label{def:tri}
\bY=\left[\begin{matrix}
[\bY_I]_{1,1}&[\bY_I]_{1,2}&&\vspace{-0.1cm}\\
[\bY_I]_{1,2}&[\bY_I]_{2,2}&\ddots&\vspace{-0.0cm}\\
&\hspace{-1.3cm}\ddots&\hspace{-0.6cm}\ddots&\hspace{-0.3cm}[\bY_I]_{N_I-1,N_I}\vspace{0.1cm}\\
&&\hspace{-0.9cm}[\bY_I]_{N_I-1,N_I}&\hspace{-0.3cm}[\bY_I]_{N_I,N_I}
\end{matrix} \right],
\end{equation}
i.e., $[\bY_I]_{i,j}=0$ if $|j-i|>1$. 

\subsubsection{Generalized Band-Connected RIS} \label{generalized band}
Band-connected RIS refers to the architecture where each RIS element is connected to its $q$-nearest elements, where 
\begin{equation}\label{band}
\bY_I=\left[\hspace{0.1cm}\begin{matrix}
[\bY_I]_{1,1}&\hspace{-0.2cm}\cdots&\hspace{-0.6cm}[\bY_I]_{1,q+1}&\hspace{-0.2cm}&\hspace{-0.2cm}&\hspace{-0.5cm}\\
\vdots&\ddots&\hspace{-0.8cm}&\hspace{-1cm}\ddots&\hspace{-1cm}&\vspace{-0.1cm}\\
[\bY_I]_{1,q+1}&\hspace{-0.1cm}&\ddots\hspace{0.8cm}&&\hspace{-0.6cm}[\bY_I]_{N_I-q,N_I}\\
&\hspace{-0.1cm}\hspace{-1cm}\ddots&\hspace{1cm}\ddots&\hspace{-0.2cm}&\hspace{-0.6cm}\vdots\vspace{0.1cm}\\
&\hspace{-0.2cm}&\hspace{-1cm}[\bY_I]_{N_I-q,N_I}&\hspace{-1cm}\cdots\hspace{-0.2cm}&\hspace{-0.3cm}[\bY_I]_{N_I,N_I}
\end{matrix}\hspace{-0.2cm}\right],
\end{equation}
 i.e., $[\bY_I]_{i,j}=0$ if $|j-i|>q$. The band-connected RIS can be generalized such that each vertex is connected to any $q$ of its subsequent vertices (not necessarily adjacent ones). This class can be extended further to include all architectures whose admittance matrices, upon left- and right-multiplication by permutation matrices, have the aforementioned structure. Please see \cite{graph} for details.   When $q=1$, the generalized band-connected RIS reduces to tree-connected RIS. In particular, the band-connected RIS in \eqref{band} reduces to tridiagonal RIS \eqref{def:tri}. 

Existing works have investigated architecture design for both MISO and MIMO systems under the conventional channel model $\bH_{\text{app},3}$. It has been proved that the generalized band-connected RIS with $q=2\min\{N_T,N_R,N_I/2\}-1$ is able to achieve the same channel shaping capability as fully-connected RIS in MIMO systems \cite{graph}.  In particular, tree-connected RIS is optimal in MISO systems \cite{tree}. 
Building on Proposition 1, we will show in the next subsection that the optimality of these architectures still hold under the general physics-consistent model.
\subsection{Optimal RIS Architecture under the General Physics-Consistent Model}\label{sec:optarch}
The channel shaping capability of a given RIS architecture is characterized by the set of channels attainable through tuning the admittances in its reconfigurable impedance network \cite{graph}. In this subsection, we characterize the optimality of band-connected RIS\footnote{For ease of presentation, we focus on the band-connected RIS in \eqref{band} in this subsection. The result extends to all generalized band-connected RIS discussed in Section \ref{generalized band}.} 
 under the general physics-consistent model. 

Let $\mathcal{H}_{\text{fully}}$ and $\mathcal{H}_{\text{band},q}$ be the channel sets achieved by fully-connected RIS and band-connected RIS with band-width $q$, respectively. 
Based on Proposition \ref{pro:general}, $\mathcal{H}_{\text{fully}}$ and $\mathcal{H}_{\text{band},q}$ can be expressed as 
$$
\begin{aligned}\mathcal{H}_{\text{fully}}=&\left\{\eqref{general_simplified}\mid(\bar{\bthe}, \bar{\bB}_I, \bB_I) \text{ satisfy  \eqref{bartheta} and \eqref{barBI}},\right.\\
&\hspace{3cm}\left.~ \bB_I\in\R^{N_I\times N_I},~~\bB_I=\bB_I^T\right\}
\end{aligned}$$
and 
$$
\begin{aligned}
\mathcal{H}_{\text{band},q}&=\left\{\eqref{general_simplified}\mid(\bar{\bthe}, \bar{\bB}_I, \bB_I) \text{ satisfy \eqref{bartheta} and  \eqref{barBI}},\right.\\
&\hspace{-0.5cm}\left.\bB_I\in\R^{N_I\times N_I},~\bB_I=\bB_I^T,~[\bB_I]_{i,j}=0\text{ if }|j-i|>q\right\}.
\end{aligned}$$
The following proposition establishes the optimality of band-connected RIS in terms of channel shaping capacity under the general physics-consistent model. 
\begin{proposition}\label{pro:arch} 
Let $D=\min\{N_R,N_T\}$, which is the multiplexing gain/degree of freedom (DoF)  of the MIMO channel. The following result holds for $q=2\min\{D, N_I/2\}-1$: 
$$\mathcal{H}_{\text{fully}}=\mathcal{H}_{\text{band},q}\cup\mathcal{N}, $$
where $\mathcal{N}$ is a low-dimensional subspace of  $\mathcal{H}_{\text{fully}}$ defined in \eqref{NU}.
\end{proposition}
\begin{proof}
The proof is similar to that of \cite[Theorem 2]{graph}. See Appendix \ref{app:proarch} for details. 
\end{proof}

Proposition  \ref{pro:arch} shows that the difference between $\mathcal{H}_{\text{fully}}$ and $\mathcal{H}_{\text{band},q}$, i.e., the achievable channel sets of fully-connected and band-connected RIS, lies only in a low-dimensional subspace of $\mathcal{H}_{\text{fully}}$, which is negligible. Intuitively, this indicates that band-connected RIS achieves the same channel shaping capability as fully-connected RIS in MIMO systems, except for a very limited number of channel realizations, which will practically never occur.  The circuit complexity, measured by the number of required admittances, in the optimal  band-connected RIS is $\mathcal{O}(N_ID)$, which is significantly lower than the $\mathcal{O}(N_I^2)$ complexity of fully-connected RIS (given that $N_I\gg D$) \cite{graph}. As a special case of Proposition \ref{pro:arch} with $N_R=1$, tree-connected RIS is optimal for MISO systems. 

A similar result to Proposition \ref{pro:arch} has been established in \cite{graph} under the approximate channel model $\bH_{\text{app,3}}$. 
In addition,  tree-connected RIS has been proved optimal for MISO systems while accounting for mutual coupling among RIS elements, i.e., under $\bH_{\text{app,2}}$, in \cite{mutualcoupling2}. Our result in Proposition \ref{pro:arch} unifies and generalizes these results by establishing the optimality of band-connected RIS for MIMO systems under the general physics-consistent channel model, which accounts for  imperfect matching, mutual coupling, and does not rely on the unilateral approximation.

\begin{figure}
\includegraphics[scale=0.33]{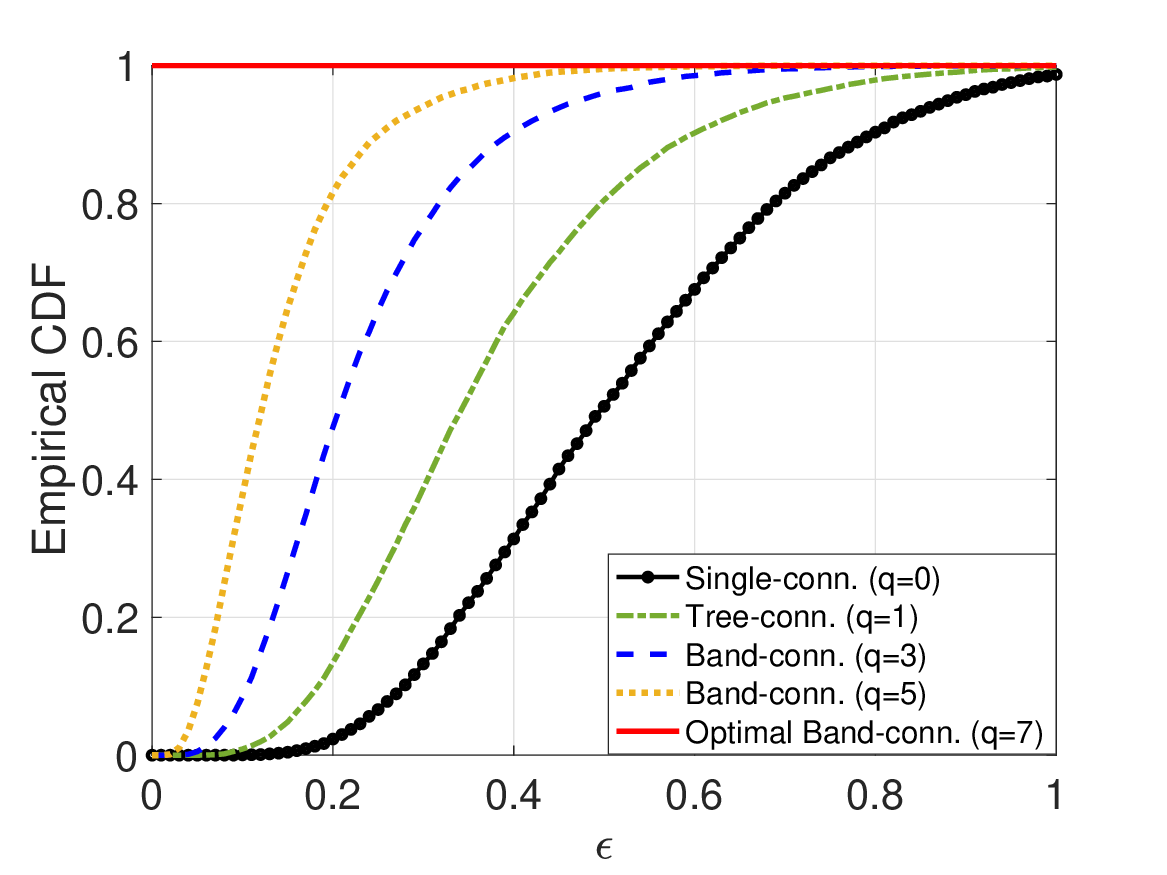}
\centering
\caption{The curves of empirical culmultive distribution function (CDF) of relative approximation error defined in \eqref{cdf}  for band-connected RIS with different band width $q$, where $N_T=N_R=4$.}
\label{J_optarch}
\end{figure} In Fig. \ref{J_optarch}, we validate Proposition \ref{pro:arch} through numerical simulations. We randomly generate $N_{\text{sample}}=10^4$ channel realizations for fully-connected RIS, where the impedance matrix $\bZ$, the source impedance  $\bZ_T$, and the load impedance $\bZ_R$ are modeled as in Section \ref{sec:simulation}. Each channel realization $\bH_{\text{fully}}^{(i)}\in\mathcal{H}_{\text{fully}}$ is obtained by randomly generating the susceptance matrix $\bB_I\in\R^{N_I\times N_I}$.   For each $\bH_{\text{fully}}^{(i)}$, we approximate it by band-connected RIS as follows\footnote{
In our simulation, we do not solve  \eqref{approx:Hfully} exactly, as it is a non-convex problem and does not admit an explicit solution. For computational tractability, we instead solve the following approximate problem:
$$
\begin{aligned}
\min_{(\bar{\bB}_I,\bB_I)}~&\|(Y_0\mathbf{I}-\mathrm{i}\bar{\bB}_I)\bar{\bH}_{IT}-(Y_0\mathbf{I}+\mathrm{i}\bar{\bB}_I)\bar{\bthe}_{\text{fully}}\bar{\bH}_{IT}\|_F^2\\
\text{s.t. }~~~& \eqref{barBI},~\bB_I=\bB_I^T,~[\bB_I]_{i,j}=0\text{ if }|j-i|>q,
\end{aligned}
$$
where we have expressed $\bH_{\text{fully}}^{(i)}$ as $\bH_{\text{fully}}^{(i)}=\bar{\bH}_{RT}+\bar{\bH}_{RI}\bar{\bthe}_{\text{fully}}\bar{\bH}_{IT}$. 
The above problem can be transformed into an unconstrained convex quadratic problem and admits closed-form solutions; see  \eqref{updateB2} and the subsequent discussions.}
:
\begin{equation}\label{approx:Hfully}
\bH_{\text{band},q}^{(i)}=\arg\min_{\bH\in\mathcal{H}_{\text{band,q}}}\|\bH-\bH_{\text{fully}}^{(i)}\|_F^2.
\end{equation}
 We quantify the approximation accuracy using the empirical cumulative distribution function (CDF) of the relative error, defined as 
 \begin{equation}\label{cdf}
F_q(\epsilon)=\frac{1}{N_{\text{sample}}}\sum_{i=1}^{N_{\text{sample}}} \mathbf{1}\left({\|\bH_{\text{band},q}^{(i)}-\bH_{\text{fully}}^{(i)}\|_F}/{\|\bH_{\text{fully}}^{(i)}\|_F}\leq \epsilon\right),
\end{equation}
where $\mathbf{1}(\cdot)$ denotes the indicator function, equal to $1$ if the condition holds and $0$ otherwise. 
Fig. \ref{J_optarch} depicts the curves of $F_{q}(\epsilon)$ for different $q$ values. 
As expected, $F_q(\cdot)$ increases with $q$.  In particular, when $q$ takes the optimal value from Proposition \ref{pro:arch}, $F_q(\epsilon)=1$ for all $\epsilon\geq 0$. This demonstrates that, with probability one, any channel achievable by a fully-connected RIS can be exactly reproduced using a band-connected RIS (with $q$ given in Proposition \ref{pro:arch}), thereby validating Proposition \ref{pro:arch}.

Before concluding this section, we remark that if the MIMO channel transmits below its full DoF, the  interconnections in the optimal architecture can further be reduced.  Specifically, let $\hat{D}$ be the number of transmit data streams, the band width in the optimal band-connected RIS reduces to $q=2\min\{\hat{D},N_I/2\}-1$; see \cite[Remark 5]{graph}.
\section{RIS Optimization under the General Physics-Consistent Model}\label{sec:opt}
In this section, we discuss optimization of RIS for different systems under the general physics-consistent model in \eqref{general}.  We begin with two simple cases: SISO systems and single-user MIMO systems with a single stream in Section \ref{sec:optimization_SISO} and Section \ref{sec:optimization_MIMO}, respectively. For both cases, we show that global optimal solutions can be obtained with fully/tree-connected RIS. 
 We then proceed to the more general case of multiuser MIMO systems in Section  \ref{sec:optimization_MUMIMO}.

\vspace{-0.1cm}
\subsection{Optimization of SISO Systems}\label{sec:optimization_SISO}
The channel gain maximization  for RIS-aided SISO systems can be formulated as 
\begin{equation}\label{SISO:problem}
\begin{aligned}
\max_{\bar{\bthe},\bar{\bB}_I,\bB_I}~&|\bar{\mathbf{h}}_{RT}+\bar{\mathbf{h}}_{RI}\bar{\bthe}\bar{\mathbf{h}}_{IT}|^2\\
\text{s.t. }\quad&\text{\eqref{bartheta} and \eqref{barBI}},~\bB_I\in\mathcal{B},
\end{aligned}
\end{equation}
where the set $\mathcal{B}$ captures the architecture of RIS; in particular,  $\mathcal{B}=\{\bB\in\R^{N_I\times N_I}\mid\bB=\bB^T\}$ for fully-connected RIS, and $\mathcal{B}=\{\bB\in\R^{N_I\times N_I}\mid\bB=\bB^T,~[\bB]_{i,j}=0,~|j-i|>1\}$ for tree-connected RIS (here we consider tridiagonal RIS as an example). The above problem is in the same form as \cite[Eqs. (42)\,--\,(43)]{mutualcoupling2}, which admits closed-form solutions 
for both fully- and tree-connected RIS, as shown in \cite{mutualcoupling2}.
\vspace{-0.1cm}
\subsection{Optimization of Single-Stream MIMO Systems}\label{sec:optimization_MIMO}
Consider a single-user MIMO system with a single stream. Let $\mathbf{w}\in\C^{N_T\times 1}$ and $\mathbf{g}^H\in\C^{1\times N_R}$ be the normalized  precoder and combiner, respectively, i.e., $\|\mathbf{w}\|_2=\|\mathbf{g}\|_2=1$. We consider the following receive power maximization problem: 
\begin{equation}\label{MIMO:problem}
\begin{aligned}
\max_{\bar{\bthe},\bar{\bB}_I,\bB_I,\mathbf{w},\mathbf{g}}~&P_T\left|\mathbf{g}^H\left(\bar{\mathbf{H}}_{RT}+\bar{\mathbf{H}}_{RI}\bar{\bthe}\bar{\mathbf{H}}_{IT}\right)\mathbf{w}\right|^2\\
\text{s.t. }\quad~~&\text{\eqref{bartheta} and \eqref{barBI}},~\bB_I\in\mathcal{B},\\
&\|\bw\|_2=1,~\|\mathbf{g}\|_2=1,
\end{aligned}
\end{equation}
where $P_T$ is the transmit power at the BS. 
The above problem has been considered in \cite{closeform} under the conventional channel model $\bH_{\text{app,3}}$, 
where two different cases have been investigated. If ${\bH}_{RT}=\mathbf{0}$ in $\bH_{\text{app,3}}$, i.e., the direct link is blocked, the  optimal solution can be obtained in closed-form. Otherwise, an iterative algorithm has been proposed to alternately update $(\mathbf{w},\mathbf{g})$ and $\bthe$ until convergence; however, the resulting solution is not guaranteed to be globally optimal.

For the general  physics-consistent channel model given by Proposition \ref{pro:general}, $\bar{\bH}_{RT}$ has a more complicated expression and is generally nonzero. As a result, closed-form solutions are generally unavailable. Nevertheless, we develop in the following a novel algorithm that is able to attain the  global optimal solution of \eqref{MIMO:problem} (for fully- and tree-connected RIS).  This is highly nontrivial due to the nonconvex nature of the problem, in particular the presence of complicated constraints and the multiplicative coupling of the variables in the objective function. As discussed earlier, the existing approach in \cite{closeform} cannot provide such a global optimality guarantee.  

  
Specifically, we first  eliminate variable $
\mathbf{g}$ by substituting the optimal $\mathbf{g}$ into the objective function of \eqref{MIMO:problem}.  By the Cauchy–Schwarz inequality, the objective function of \eqref{MIMO:problem} is maximized with respect to  $\mathbf{g}$ when $\mathbf{g}$  is aligned with $\left(\bar{\mathbf{H}}_{RT}+\bar{\mathbf{H}}_{RI}\bar{\bthe}\bar{\mathbf{H}}_{IT}\right)\mathbf{w}$, i.e.,
$$\mathbf{g}^*(\bthe,\bw)=\frac{\left(\bar{\mathbf{H}}_{RT}+\bar{\mathbf{H}}_{RI}\bar{\bthe}\bar{\mathbf{H}}_{IT}\right)\mathbf{w}}{\|\left(\bar{\mathbf{H}}_{RT}+\bar{\mathbf{H}}_{RI}\bar{\bthe}\bar{\mathbf{H}}_{IT}\right)\mathbf{w}\|_2}.$$ Substituting $\mathbf{g}^*(\bthe,\bw)$  into the objective function of \eqref{MIMO:problem}  yields the following equivalent problem:  
\begin{equation}\label{MIMO:problem2}
\begin{aligned}
\max_{\bar{\bthe},\bar{\bB}_I,\bB_I, \bw}~&\|\left(\bar{\mathbf{H}}_{RT}+\bar{\mathbf{H}}_{RI}\bar{\bthe}\bar{\mathbf{H}}_{IT}\right)\bw\|_2^2\\
\text{s.t. }\quad&\text{\eqref{bartheta} and \eqref{barBI}},~\bB_I\in\mathcal{B},\\
&\|\bw\|_2=1.
\end{aligned}
\end{equation}
The constant $P_T$ does not affect the optimal solution and is thus omitted for brevity. 
To further handle the multiplicative coupling between the variables $\bar{\bthe}$ and $\bw$ in the objective function, we introduce an auxiliary variable $\mathbf{u}=\bar{\bthe}\bar{\mathbf{H}}_{IT}\bw$, then problem \eqref{MIMO:problem2} transforms to
\begin{equation}\label{MIMO:problem3}
\begin{aligned}
\max_{\bar{\bthe},\bar{\bB}_I,\bB_I, \u,\mathbf{w}}~&\|\bar{\mathbf{H}}_{RT}\mathbf{w}+\bar{\mathbf{H}}_{RI}\mathbf{u}\|^2_2\\
\text{s.t. }\qquad&\text{\eqref{bartheta} and \eqref{barBI}},~\bB_I\in\mathcal{B},\\
&\mathbf{u}=\bar{\bthe}\bar{\mathbf{H}}_{IT}\bw,~~\|\bw\|_2=1.
\end{aligned}
\end{equation}
Our approach for solving \eqref{MIMO:problem3} contains two steps. First, we solve the following relaxation model of \eqref{MIMO:problem3}:
\begin{equation}\label{MIMO:problem4}
 \begin{aligned}
\max_{\mathbf{u},\bw}~&\|\bar{\mathbf{H}}_{RT}\mathbf{w}+\bar{\mathbf{H}}_{RI}\mathbf{u}\|^2_2\\
\text{s.t. }~&\|\mathbf{u}\|_2=\|\bar{\mathbf{H}}_{IT}\bw\|_2,~\|\bw\|_2=1.
\end{aligned}
\end{equation}
In \eqref{MIMO:problem4}, we relax all constraints in \eqref{MIMO:problem3} related to $(\bar{\bthe},\bar{\bB}_I, {\bB}_I)$, and introduce the constraint $\|\mathbf{u}\|_2=\|\bar{\mathbf{H}}_{IT}\bw\|_2$, which holds since $\bar{\bthe}$ is unitary.  The optimal value of \eqref{MIMO:problem4}, denoted by $P_u$, serves as an upperbound of that of \eqref{MIMO:problem3}, denoted by $P^*$, i.e., $P^*\leq P_u$. Let $(\mathbf{u}^*,\bw^*)$ be the optimal solution to \eqref{MIMO:problem4}. The second step is to recover $(\bar{\bthe},\bar{\bB}_I,\bB_I)$ by taking back related constraints and solving the following system:
\begin{equation}\label{solve:barBI}
\left\{
\begin{aligned}
&\mathbf{u}^*=\bar{\bthe}\bar{\mathbf{H}}_{IT}\bw^*,\\
&\text{\eqref{bartheta} and  \eqref{barBI}},~\bB_I\in\mathcal{B}.
\end{aligned}\right.
\end{equation}
Since $\|\mathbf{u}^*\|_2=\|\bar{\mathbf{H}}_{IT}\bw^*\|_2$, \eqref{solve:barBI} admits a unique solution $(\bar{\bthe}^*,\bar{\bB}_I^*,\bB_I^*)$ for tree-connected RIS; see \cite[Eqs. (45)\,--\,(51)]{mutualcoupling2}.  Therefore, $(\bar{\bthe}^*, \bar{\bB}_I^*, \bB_I^*, \mathbf{u}^*, \bw^*)$ is feasible for problem \eqref{MIMO:problem} and achieves the upperbound objective value $P_u$, which is thus an optimal solution to \eqref{MIMO:problem}.  


Based on the above discussions, the remaining task is to solve \eqref{MIMO:problem4}. We employ the SDR approach. Specifically, let $$\begin{aligned}
\mathbf{Q}_0&=\left[\begin{matrix} \bar{\bH}_{RT}^H\bar{\bH}_{RT}&\bar{\bH}_{RT}^H\bar{\bH}_{RI}\\\bar{\bH}_{RI}^H\bar{\bH}_{RT}&\bar{\bH}_{RI}^H\bar{\bH}_{RI}\end{matrix}\right],\\
\mathbf{Q}_1&=\left[\begin{matrix}\bar{\bH}_{IT}^H\bar{\bH}_{IT}&\\&-\mathbf{I}_{N_I}\end{matrix}\right],~\mathbf{Q}_2=\left[\begin{matrix}\mathbf{I}_{N_T}&\\&\mathbf{0}\end{matrix}\right], 
\end{aligned}$$
$\x=[\bw^T,\mathbf{u}^T]^T\in\C^{(N_T+N_I)\times 1}$,  and $\bX=\x\x^H$. 
The semidefinite relaxation of \eqref{MIMO:problem4} is
\begin{equation}\label{SDP}
\begin{aligned}
\max_{\bX\succeq \mathbf{0}}~~&\text{tr}(\mathbf{Q}_0\bX)\\
\text{s.t.}~~~&\text{tr}(\mathbf{Q}_1\bX)=0,~\text{tr}(\mathbf{Q}_2\bX)=1,
\end{aligned}
\end{equation}
where the non-convex rank-one constraint $\text{rank}(\bX)=1$ is relaxed.  A well-known result  for a complex SDP of the form \eqref{SDP} is that, there exists an optimal solution $\mathbf{X}^\star$ with $\operatorname{rank}(\mathbf{X}^\star) \le \sqrt{m}$ (as long as the problem is feasible), where $m$ denotes the number of linear constraints  \cite{Huang2010SDP,SDR}.  Since $m=2$ in \eqref{SDP}, there exists a rank-one optimal solution, and thus the SDR is tight. By decomposing $\bX^*=\x^*{\x^*}^H$, we get the optimal solution to \eqref{MIMO:problem4}.

\begin{remark}[A Low-Dimensional SDR]
Solving the ($N_T+N_I$)-dimensional semidefinite program (SDP) in \eqref{SDP} is computationally expensive, as the number of RIS elements is typically large. This remark will show that, by carefully exploiting the problem structure, it suffices to solve an SDP with a much lower dimension of $N_T+N_R$, where $N_R \ll N_I$ in practice.

Let $\bar{\bH}_{RI}=\bU_{RI}\bD_{RI}\mathbf{V}_{RI}^H$ be the singular value decomposition of $\bar{\bH}_{RI}$, where the singular values are sorted in descending order. Then,  at most the first $N_R$ diagonal enetries in $\bD_{RI}$ are non-zero. 
Introducing an auxiliary variable $\hat{\mathbf{u}}=\mathbf{V}_{RI}^H\mathbf{u}\in\C^{N_I\times 1}$, problem \eqref{MIMO:problem4} transforms to 
\begin{equation}\label{MIMO:problem5}
 \begin{aligned}
\max_{\hat{\mathbf{u}},\bw}~&\|\bU_{RI}^H\bar{\mathbf{H}}_{RT}\mathbf{w}+{\mathbf{D}}_{RI}\hat{\mathbf{u}}\|_2^2\\
\text{s.t. }~&\|\hat{\mathbf{u}}\|_2=\|\bar{\mathbf{H}}_{IT}\bw\|_2,~\|\bw\|_2=1.
\end{aligned}
\end{equation}
We claim that the optimal solution to \eqref{MIMO:problem5} satisfies 
\begin{equation}\label{hatustar}
[\hat{\mathbf{u}}^*]_i=0,~\forall~i>N_R.
\end{equation}
To prove this claim, we note that $\mathbf{D}_{RI} \hat{\mathbf{u}}$ is nonzero only in its first $N_R$ entries. In addition, the constraint on $\hat{\mathbf{u}}$ is independent of its phase. Therefore, to maximize the objective function, the phase of each $[\hat{\mathbf{u}}]_i$ should be adjusted such that the phases of $[\mathbf{U}_{RI}^H \bar{\mathbf{H}}_{RT} \mathbf{w}]_i$ and $[\mathbf{D}_{RI}]_{i,i} [\hat{\mathbf{u}}]_i$ are aligned. Under this phase alignment, the objective function becomes
$$\|\bU_{RI}^H\bar{\bH}_{RT}\bw\|_2^2+\sum_{i=1}^{N_R}|[\bD_{RI}]_{i,i}[\hat{\mathbf{u}}]_i|^2,$$
which is strictly increasing in $|[\hat{\mathbf{u}}]_i|$ for $i \leq N_R$ and is independent of $[\hat{\mathbf{u}}]_i$ for $i > N_R$. Hence, \eqref{hatustar} holds; otherwise, one could achieve a larger objective value by increasing the magnitudes of the first $N_R$ elements of $\hat{\mathbf{u}}$ and setting its last $N_I-N_R$ elements to zero.

With \eqref{hatustar}, we only need to optimize the first $N_R$ entries of $\hat{\mathbf{u}}$. Let $\bar{\mathbf{u}}=[\hat{\mathbf{u}}]_{1:N_R}$ and $\bar{\bD}_{RI}=[\bD_{RI}]_{:,1:N_R}.$  Utilizing \eqref{hatustar}, problem \eqref{MIMO:problem5} reduces to the following:
\begin{equation}\label{MIMO:problem6}
 \begin{aligned}
\max_{\bar{\mathbf{u}},\bw}~&\|\bU_{RI}^H\bar{\mathbf{H}}_{RT}\mathbf{w}+\bar{\mathbf{D}}_{RI}\bar{\mathbf{u}}\|_2^2\\
\text{s.t. }~&\|\bar{\mathbf{u}}\|_2=\|\bar{\mathbf{H}}_{IT}\bw\|_2,~\|\bw\|_2=1.
\end{aligned}
\end{equation}
The above problem can still be solved via SDR, and the resulting SDP has a dimension of $N_T + N_R$, which is significantly lower than that of \eqref{SDP}.
\end{remark}
\subsection{Optimization of Multiuser MIMO Systems}\label{sec:optimization_MUMIMO}
Finally, we discuss sum-rate maximization for  multiuser MIMO systems. For ease of presentation, we assume that each user has a single antenna. In this case, $N_R$ is the number of users in the system. The sum-rate maximization problem can be formulated as 
\begin{equation}\label{multiuser}
\begin{aligned}
\max_{\mathbf{W},\bar{\bthe},\bar{\bB}_I,\bB_I}~&\sum_{k=1}^{N_R}\log\left(1+\frac{|\bh_k(\bar{\bthe})^H\bw_k|^2}{\sum_{j\neq k}|\bh_k(\bar{\bthe})^H\bw_j|^2+\sigma^2}\right)\\
\text{s.t.}~~~~~~ &\text{\eqref{bartheta} and \eqref{barBI}},~\bB_I\in\mathcal{B},\\
&\|\mathbf{W}\|_F^2\leq P_T,
\end{aligned}
\end{equation}
where $\h_k(\bar{\bthe})^H$ is the $k$-th row of $\bH(\bar{\bthe})=\bar{\mathbf{H}}_{RT}+\bar{\mathbf{H}}_{RI}\bar{\bthe}\bar{\mathbf{H}}_{IT}$. 

Various algorithms have been proposed to solve the sum-rate maximization problem under the conventional channel model. The new challenge introduced by the general physics-consistent model lies in that the virtual scattering matrix $\bar{\bthe}$ involved in the channel model is related to the actual susceptance matrix $\bB_I$ through a more intricate relationship, given by \eqref{bartheta} and \eqref{barBI}. For fully-connected RIS, this constraint reduces to the unitary and symmetric constraints on $\bar{\bthe}$, and thus existing approaches remain directly applicable. However, for more general architectures,  \eqref{bartheta} and \eqref{barBI} must be explicitly taken into account. In particular, for group-connected RIS, the block-diagonal structure of the susceptance matrix $\bB_I$ does not generally translate into a block-diagonal structure of $\bar{\bthe}$, making existing approaches that rely on such structural properties not directly applicable.

Among the existing methods, our recent work in \cite{wu} proposed an ADMM-based optimization framework that is applicable to arbitrary RIS architectures, which was shown to achieve a favorable trade-off between computational complexity and performance. In the following, we generalize this algorithmic framework to accommodate the general physics-consistent channel model, by carefully dealing with the coupling between $\bar{\bB}_I$ and $\bB_I$ in \eqref{barBI}.

To encompass arbitrary RIS architectures, we introduce the following notation for the architecture-specific set $\mathcal{B}$: 
$$\mathcal{B}=\{\bB_I\in\R^{N_I\times N_I}\mid \bB_I=\bB_I^T,~[\bB_{I}]_{i,\mathcal{S}_i}=\mathbf{0}\},$$
where $\mathcal{S}_i$ defines the interconnection pattern of the $i$-th RIS element. In particular,  $j\in\mathcal{S}_i$ if and only if there is no interconnection between the $i$-th and $j$-th RIS elements.

As in \cite{wu}, we introduce an auxiliary variable $\bU=(\bar{\bH}_{RI}\bar{\bthe})^{H}$ to reformulate \eqref{multiuser} as
\begin{subequations}\label{multiuser2}
\begin{align}
\max_{\mathbf{W},\bU,\bar{\bB}_I,\bB_I}&\sum_{k=1}^{N_R}\log\left(1+\frac{|\bh_k(\bU)^H\bw_k|^2}{\sum_{j\neq k}|\bh_k(\bU)^H\bw_j|^2+\sigma^2}\right)\\
\text{s.t.}~~~~ &\left(Y_0\mathbf{I}-\mathrm{i}\bar{\mathbf{B}}_I\right)\bU=\left(Y_0\mathbf{I}+\mathrm{i}\bar{\mathbf{B}}_I\right)\bar{\bH}_{RI},\label{29b}\\
~~&\eqref{barBI},~ \bB_I\in\mathcal{B},\\
&\|\mathbf{W}\|_F^2\leq P_T,
\end{align}
\end{subequations}
where $\bh_{k}(\bU)^H$ is the $k$-th row of  $\bH(\bU)=\bar{\mathbf{H}}_{RT}+\bU^H\bar{\mathbf{H}}_{IT}.$ 
To solve problem \eqref{multiuser2}, we employ the ADMM-based algorithmic framework  in \cite[Eqs. (11a)\,--\,(11f)]{wu},  which first simplifies the sum-rate objective function using fractional programming technique \cite{FP,FP2} and  penalizes the bilinear constraint \eqref{29b} into the objective function via the augmented Lagrangian (AL) function, and then alternately updates the variables until convergence. 
The only difference lies in the update of $(\bar{\bB}_I, \bB_I)$, where  the architecture-specific constraint on $\bB_I$, i.e., $\bB_I \in \mathcal{B}$, is coupled with $\bar{\bB}_I$ through the linear constraint \eqref{barBI}.  In the following, we focus on the update of $(\bar{\bB}_I,\bB_I)$ involved in the ADMM algorithm. Details on the updates for the other variables are omitted for brevity and can be found in \cite{wu}.  

Analogous to \cite[Eq. (16)]{wu}, the $(\bar{\bB}_I,\bB_I)$-subproblem at the $(t+1)$-th iteration is
\begin{equation}\label{updateB1}
\begin{aligned}
\min_{\bar{\bB}_I,\bB_I}\hspace{-0.05cm}&~\frac{\rho}{2}\left\|\bar{\bB}_I(\mathrm{i}\bU^t\hspace{-0.05cm}+\hspace{-0.05cm}\mathrm{i}\bar{\H}_{RI})\hspace{-0.05cm}-\hspace{-0.05cm}\left(Y_0\bU^t\hspace{-0.05cm}-\hspace{-0.05cm}Y_0\bar{\H}_{RI}\hspace{-0.05cm}+\hspace{-0.05cm}\frac{\blam^{t}}{\rho}\right)\right\|^2\hspace{-0.05cm}\\
&\hspace{4.5cm}+\frac{\xi}{2}\|\bar{\bB}_I-\bar{\bB}_I^t\|^2_F\\
\text{s.t.}\,~&~\eqref{barBI},~ \bB_I\in\mathcal{B},
\end{aligned}
\end{equation}
where $\boldsymbol{\lambda}$ and $\rho$ are, respectively,  the Lagrange multiplier and penalty parameter in the AL function associated with the linear constraint \eqref{29b}, and $\xi$ is the regularization parameter. By transforming problem \eqref{updateB1} into the real space and eliminating variable $\bar{\bB}_I$ by constraint \eqref{barBI}, we obtain the following problem on $\bB_I$: 
 \begin{equation}\label{updateB2}
\min_{\bB_I\in\mathcal{B}}~~\frac{\rho}{2}\|\mathbf{L}\bB_I\mathbf{R}-\mathbf{\Gamma}_1\|_F^2+\frac{\xi}{2}\|\mathbf{L}\bB_I\mathbf{L}-\mathbf{\Gamma}_2\|_F^2,
\end{equation}
where $\mathbf{L}=\sqrt{Y}_0\RR(\bar{\bY}_{II})^{-\frac{1}{2}}\in\R^{N_I\times N_I},$ 
$$\begin{aligned}
\mathbf{R}&=\mathbf{L}[\RR(\mathrm{i}\bU^t+\mathrm{i}\bar{\H}_{RI}),~\I(\mathrm{i}\bU^t+\mathrm{i}\bar{\H}_{RI})]\in\R^{N_I\times 2N_R}\\
\boldsymbol{\Gamma}_1&=-\mathbf{L}\,\mathcal{I}(\bar{\bY}_{II})\mathbf{R}+\hspace{-0.05cm}\left[\RR\hspace{-0.1cm}\left(Y_0(\bU^t\hspace{-0.05cm}-\hspace{-0.05cm}\bar{\H}_{RI})\hspace{-0.05cm}+\hspace{-0.05cm}\frac{\blam^{t}}{\rho}\right)\right.,\\
&\hspace{2.7cm}\left.\I\hspace{-0.05cm}\left(Y_0(\bU^t\hspace{-0.05cm}-\hspace{-0.05cm}\bar{\H}_{RI})\hspace{-0.05cm}+\hspace{-0.05cm}\frac{\blam^{t}}{\rho}\right)\right]\in\R^{N_I\times 2N_R},
\end{aligned}
$$
and 
$\boldsymbol{\Gamma}_2=\bar{\bB}_I^t-\mathbf{L}\I(\bar{\bY}_{II})\mathbf{L}\in\R^{N_I\times 2N_R}.$
In the following, we transform \eqref{updateB2} into an unconstrained quadratic program. 
 Let 
\begin{equation}\label{relation:xB}
\x=[[\bB_{I}]_{1,\mathcal{S}_1^c},[\bB_{I}]_{2,\mathcal{S}_2^c},\dots,[\bB_{I}]_{N_I,\mathcal{S}_{N_I}^c}]^T
\end{equation}
with $
\mathcal{S}_i^c=\{j\geq i\mid j\notin \mathcal{S}_i\}, 
$ i.e., 
  $\x$ collects all the non-zero elements in the upper tridiagonal of $\bB_I$.  In addition, let $\boldsymbol{l}_i$ denote the $i$-th column of $\mathbf{L}$, $\mathbf{r}_{i}^T$ denote the $i$-th row of $\mathbf{R}$, $\bb=\text{vec}(\boldsymbol{\Gamma}_1^T)$, and $\mathbf{d}=\text{vec}(\boldsymbol{\Gamma}_2^T)$. 
   Then problem \eqref{updateB2} can be equivalently expressed as
   \begin{equation}\label{QP}
   \min_{\x}~\frac{\rho}{2}\|\bA\x-\bb\|_2^2+\frac{\xi}{2}\|\bC\x-\mathbf{d}\|_2^2,
   \end{equation}
 where 
 $\bA=[\bA_1,\bA_2,\dots,\bA_{N_I}]$ with 
 $$
[\bA_{i}]_{:,q}=\left\{
 \begin{aligned}
   \boldsymbol{l}_{i}\otimes\mathbf{r}_{i},~~~~~~~~~~~&\text{if }q=1;\\
  \boldsymbol{l}_{i}\otimes\mathbf{r}_{i_q}+\boldsymbol{l}_{i_q}\otimes \mathbf{r}_i,~~~&\text{otherwise},\\
   \end{aligned}
   \right.
   $$
   and 
   $\bC=[\bC_1,\bC_2,\dots,\bC_{N_I}]$ with 
 $$
[\bC_{i}]_{:,q}=\left\{
 \begin{aligned}
   \boldsymbol{l}_{i}\otimes\boldsymbol{l}_{i},~~~~~~~~~~~&\text{if }q=1;\\
  \boldsymbol{l}_{i}\otimes\boldsymbol{l}_{i_q}+\boldsymbol{l}_{i_q}\otimes \boldsymbol{l}_i,~~~&\text{otherwise}.\\
   \end{aligned}
   \right.
   $$
Solving \eqref{QP} gives 
\begin{equation}\label{solution:x}
\x=(\rho\bA^T\bA+\xi\bC^T\bC)^{-1}(\rho\bA^T\bb+\xi\bC^T\mathbf{d}).
\end{equation}
By further utilizing \eqref{relation:xB} gives the solution to \eqref{updateB1}.

We remark that the problem in \cite[Eq. (17)]{wu} corresponds to the special case of \eqref{updateB2} with $\mathbf{L}=\mathbf{I}$. By substituting $\mathbf{L}=\mathbf{I}$ into matrices $\bA$ and $\bC$, the solution in \eqref{solution:x} reduces to \cite[Eq. (20)]{wu}.

The adopted framework can be readily extended to the multiuser MIMO scenario by applying the fractional programming technique in its matrix form \cite{FP3}. This generalization has been discussed in \cite[Section V-B]{wu}.

\section{Simulation}\label{sec:simulation}
In this section, we present simulation results to examine the impact of different EM effects and approximations on system performance, using the algorithms proposed in Section \ref{sec:opt}. We focus on fully-connected RIS and evaluate its performance under different channel models.  In particular, we investigate two key effects: mutual coupling (MC) among RIS elements and unilateral approximation (UA). To this end, we assume that the source impedances at the transmitter and load impedances  at the receiver  are equal to the reference impedance $Z_0$,  set as $Z_0=50 \,\Omega$, which gives  $\bZ_R=Z_0 \mathbf{I}$ and  $\bZ_T=Z_0\mathbf{I}$, and  the 
 transmit and receive antennas are perfectly matched with no mutual coupling, which gives $\bZ_{TT}=Z_0\mathbf{I}$ and $\bZ_{RR}=Z_0\mathbf{I}.$  In addition, the direct link between the transmitter and receiver is assumed to be fully obstructed, i.e., $\bZ_{RT}=\mathbf{0}$. We focus on the following three channel models:
\begin{enumerate}
\item The general physics-consistent model in \eqref{general_model}. We  denote all related lines  as ``MC aware, w/o UA'';
\item The approximate model with unilateral approximation given by  $\bH_{\text{app,2}}$, which accounts for mutual coupling among RIS elements. We  denote all lines related to $\bH_{\text{app,2}}$ as ``MC aware, w/ UA'';
\item The approximate model given by $\bH_{\text{app,3}}$, which assumes no mutual coupling among RIS elements and that the unilateral approximation holds. We  denote all lines related to $\bH_{\text{app,3}}$ as ``MC unaware, w/ UA''.

\end{enumerate}
\begin{figure*}
\subfigure[SISO system.]{\includegraphics[width=0.33\textwidth]{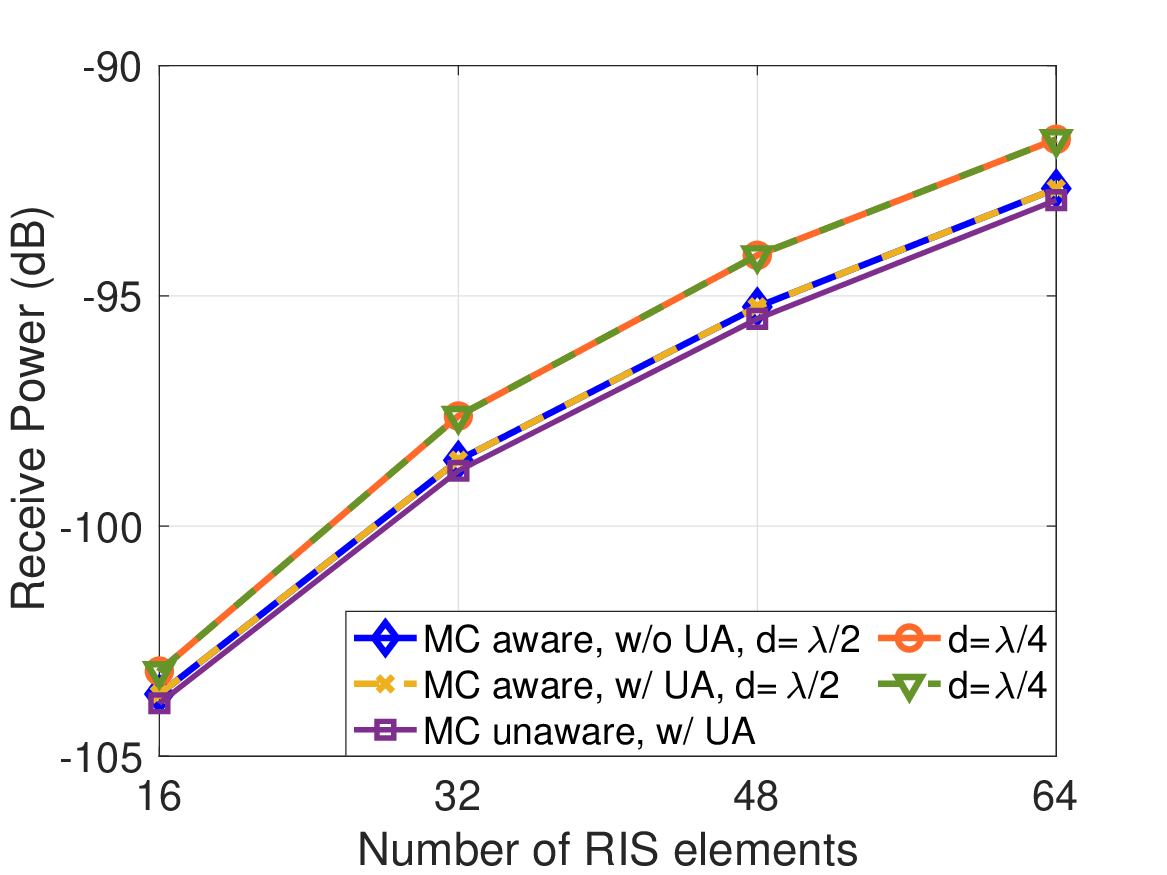}}
\subfigure[Single-stream MIMO system, $N_T=N_R=4$.]{\includegraphics[width=0.33\textwidth]{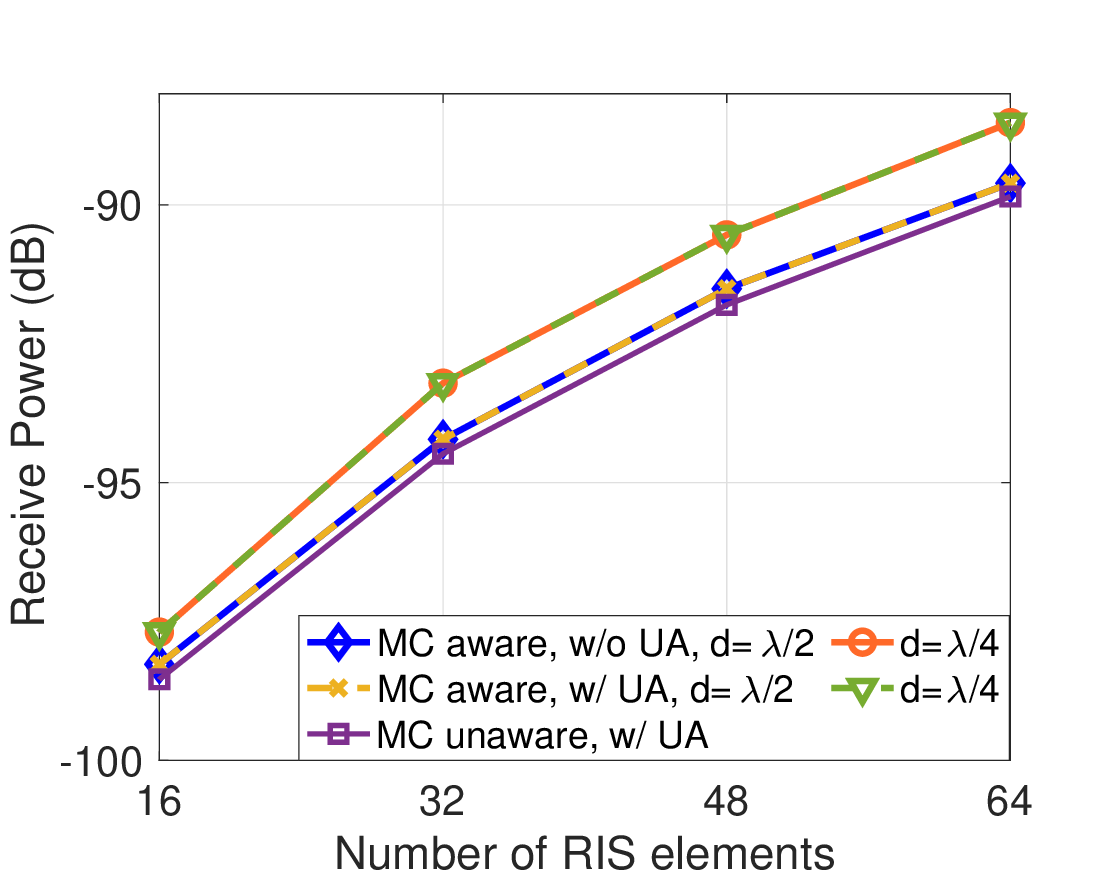}}
\subfigure[Multiuser MISO system, $N_T=N_R=4$.]{\includegraphics[width=0.33\textwidth]{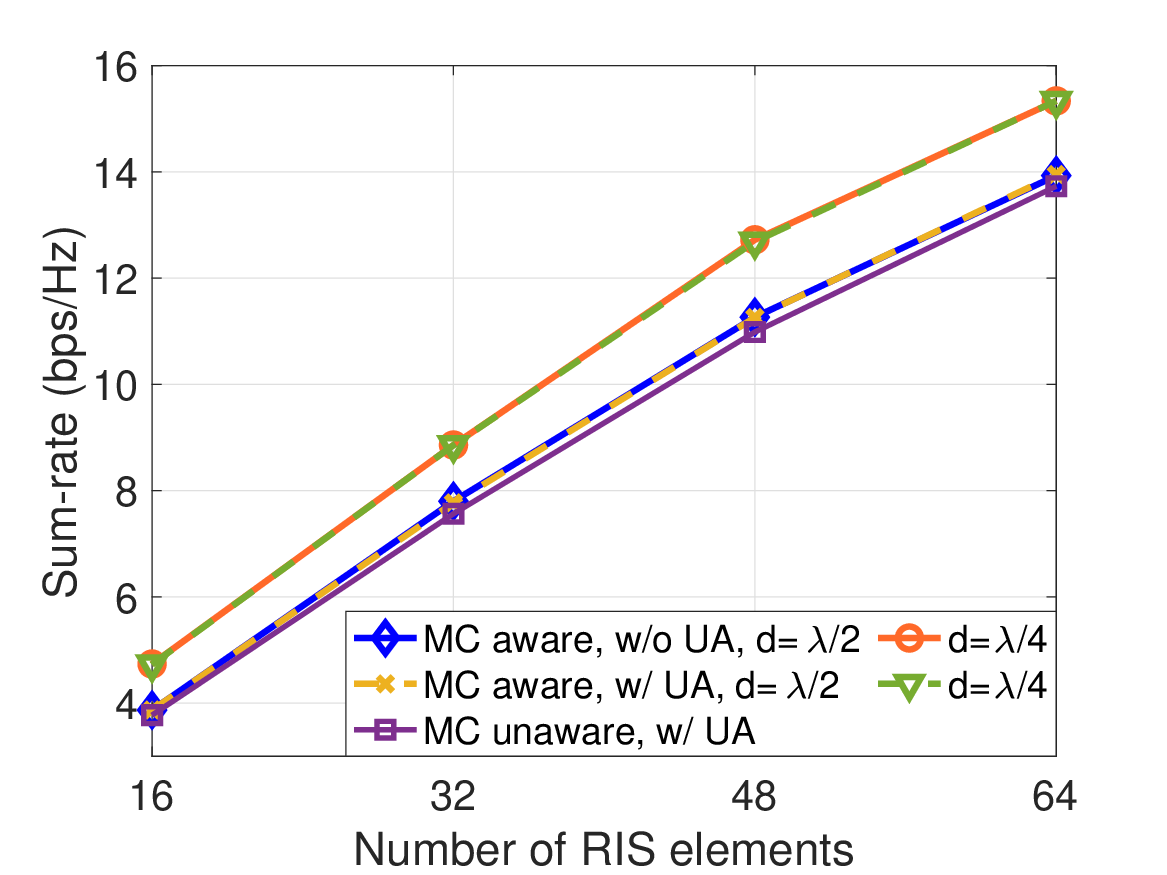}}
\caption{Receive power and sum-rate versus the number of RIS elements for different channel models and inter-element spacing $d$. }
\label{MIMO_receive}
\centering
\end{figure*}
\begin{figure*}
\subfigure[SISO system.]{\includegraphics[width=0.33\textwidth]{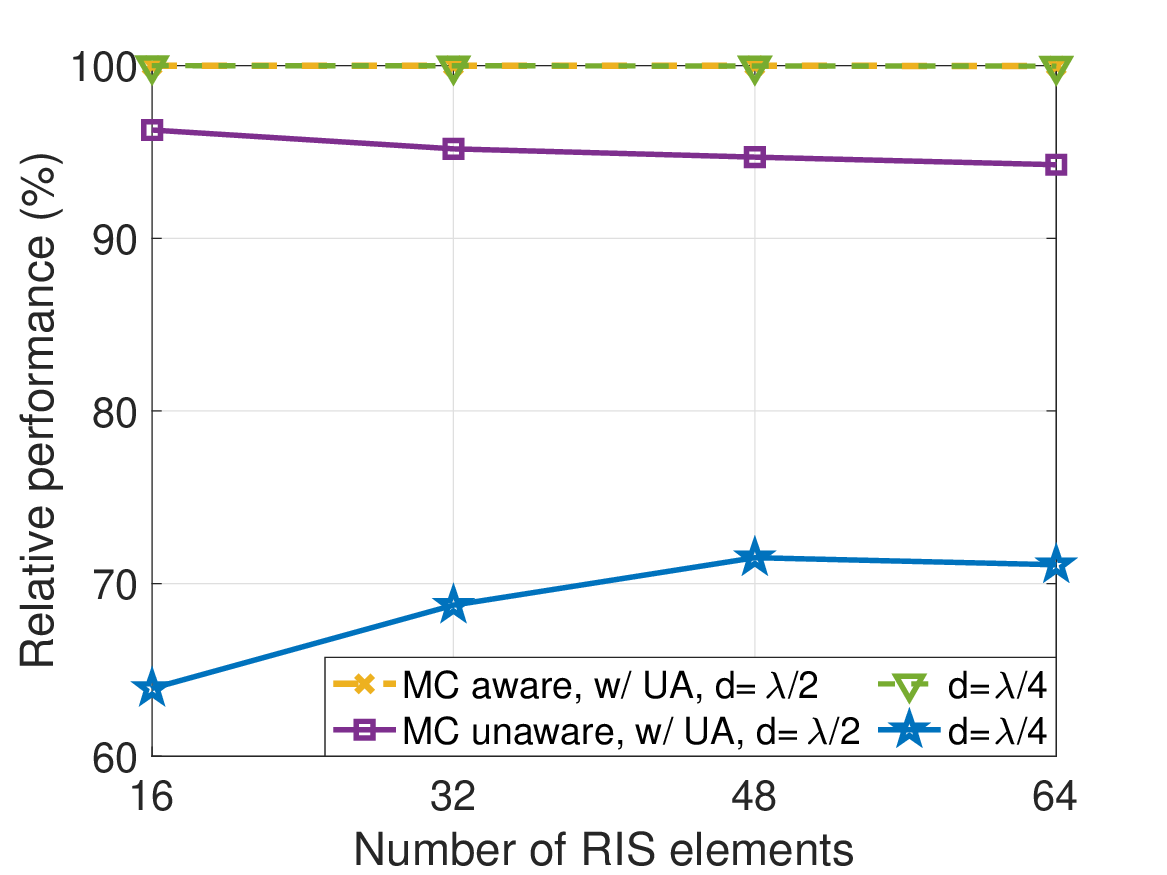}}
\subfigure[Single-stream MIMO system, $N_T=N_R=4$.]{\includegraphics[width=0.33\textwidth]{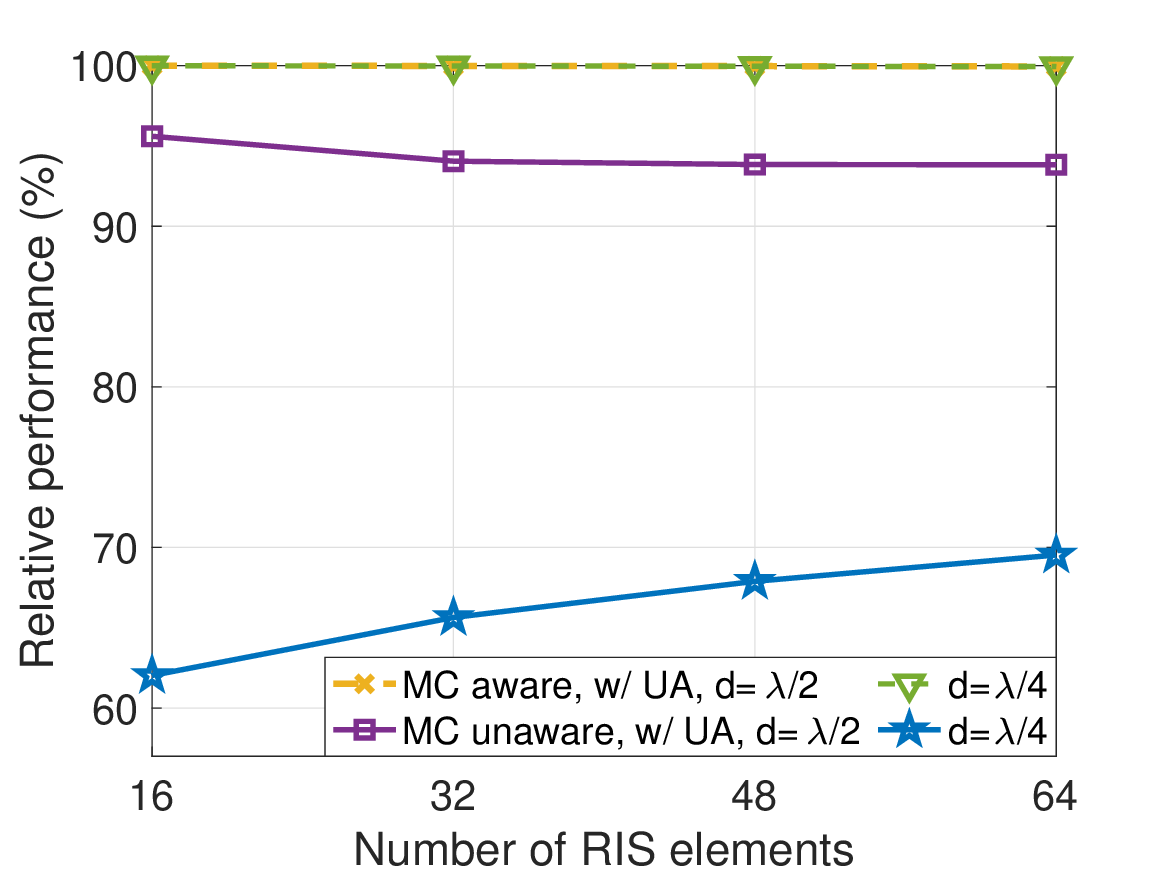}}
\subfigure[Multiuser MISO system, $N_T=N_R=4$.]{\includegraphics[width=0.33\textwidth]{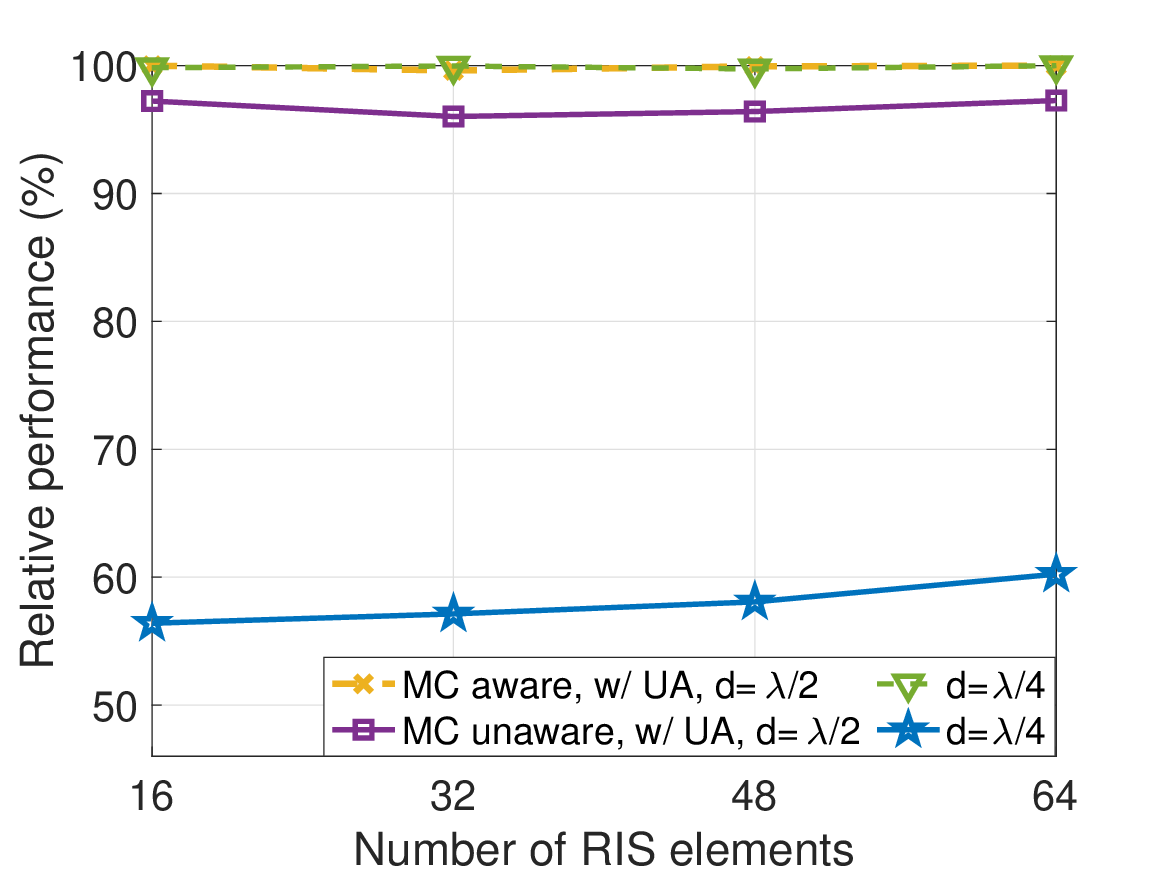}}
\caption{Relative performance of the solutions obtained from approximate channel models with respect to the solution obtained from the general physics-consistent model, which is computed as  $\frac{F(\bH(\bar{\bthe}_{\text{app}}),\mathbf{W}_{\text{app}})}{F(\bH(\bar{\bthe}^*),\mathbf{W}^*)}\times 100\%$, where $F(\bH(\bar{\bthe}),\mathbf{W})$ is the receive power (i.e., $F(\bH(\bar{\bthe}),\mathbf{W})=\|\bH(\bar{\bthe})\|_2^2$) for the SISO and single-stream MIMO systems, and is the sum-rate (i.e., $F(\bH(\bar{\bthe}),\mathbf{W})=\sum_{k=1}^{N_R}\log(1+\frac{|\bh_k(\bar{\bthe})^H\bw_k|^2}{\sum_{j\neq k}|\bh_k(\bar{\bthe})^H\bw_j|^2+\sigma^2})$) for the multiuser MISO system, $(\bar{\bthe}_{\text{app}},\mathbf{W}_{\text{app}})$ is  the solution obtained using  the approximate model $\bH_{\text{app},2}$ or $\bH_{\text{app},3}$, and $(\bar{\bthe}^*,\mathbf{W}^*)$ is the solution obtained using $\bH$.}
\label{fig_relative}
\centering
\end{figure*}

As in \cite{mutualcoupling2}, we consider a RIS implemented as a uniform planar array (UPA) of radiating elements in the $x-y$ plane, and with inter-element distance $d$. The RIS elements are thin wire dipoles parallel to the $y$ axis with length $\ell=\lambda/4$ and radius 
$r\ll \ell$, where $\lambda=c/f$ is the wavelength of the frequency $f=28$ GHz, and $c$ is the speed of light. All RIS elements are assumed to be perfectly matched to $Z_0=50\,\Omega$, i.e., $[\bZ_{II}]_{n,n}=Z_0$. Following \cite{mutualcoupling, mutualcoupling2}, we model $[\bZ_{II}]_{p,q}$ with $p\neq q$, which captures the mutual coupling among the RIS elements located at $(x_p,y_p)$ and $(x_q,y_q)$, as follows:
\begin{equation}\label{mutual}
\begin{aligned}
& {\left[\mathbf{Z}_{I I}\right]_{q, p}=\int_{y_q-\frac{\ell}{2}}^{y_q+\frac{\ell}{2}} \int_{y_p-\frac{\ell}{2}}^{y_p+\frac{\ell}{2}} \frac{\mathrm{i} \eta_0}{4 \pi k_0}\left(\frac{\left(y^{\prime \prime}-y^{\prime}\right)^2}{d_{q, p}^2}\right.} \\
& \left.\times\left(\frac{3}{d_{q, p}^2}+\frac{3 \mathrm{i} k_0}{d_{q, p}}-k_0^2\right)-\frac{\mathrm{i} k_0+d_{q, p}^{-1}}{d_{q, p}}+k_0^2\right) \frac{e^{-\mathrm{i} k_0 d_{q, p}}}{d_{q, p}} \\
& \times\frac{\sin \left(k_0\left(\frac{\ell}{2}-\left|y^{\prime}-y_p\right|\right)\right) \sin \left(k_0\left(\frac{\ell}{2}-\left|y^{\prime \prime}-y_q\right|\right)\right)}{\sin ^2\left(k_0 \frac{\ell}{2}\right)} d y^{\prime} d y^{\prime \prime},
\end{aligned}
\end{equation}
where $\eta_0 = 377 \Omega$ is the impedance of free space,  $k_0=2\pi/\lambda$
is the wavenumber, and $d_{q,p}=(x_q-x_p)^2 + (y''-y')^2$. 

Throughout the simulations, the transmit power is set to $P_T = 20$ dBm and the noise power is $\sigma^2 = -80$ dBm. To thoroughly analyze the impact of mutual coupling and unilateral approximation, we consider two scenarios with different modeling assumptions for the channels $\mathbf{Z}_{RI}$ and $\bZ_{IT}$.

\begin{enumerate}
\item First, we consider a common scenario in which the distances between the transmitter, RIS, and receiver are sufficiently large. In this case, the elements of $\bZ_{RI}$ and $\bZ_{IT}$ are generated as independent Gaussian random variables  \cite{mutualcoupling2}, i.e., $[\bZ_{RI}]_{i,j}\sim\mathcal{CN}(0,\rho_{RI}),$ and $[\bZ_{IT}]_{i,j}\sim\mathcal{CN}(0,\rho_{IT})$, with path gain $\rho_{IT}=4Z_0^210^{-8}$ and $\rho_{RI}=4Z_0^210^{-4}$. The corresponding results are given in Figs. \ref{MIMO_receive}\,--\,\ref{fig_wrtK}.
\item Then,  we consider a scenario in which the distance between the transmitter and the RIS is on the order of the wavelength. The receiver is still assumed to be located far from the RIS. Specifically,  we set the transmitter parallel to the $y$-axis, centered at $(d,0,r\lambda)$,  where $d$ is the inter-element spacing of RIS, $\lambda$ is the wavelength, and $r$ is a factor determining the distances between transmitter and  RIS.  The matrix $\bZ_{IT}$ is then computed according to the model in \eqref{mutual}, and all other components of $\bZ$  are generated as in 1). The results are given in Fig. \ref{fig_smalld}.
\end{enumerate}

To investigate the performance of different channel models for different systems, we solve \eqref{SISO:problem}, \eqref{MIMO:problem}, and \eqref{multiuser} under both the general physics-consistent model $\bH$ and the approximate models $\bH_{\text{app},2}$ and $\bH_{\text{app},3}$. Specifically, the approximate solutions are obtained by applying the algorithms in Section \ref{sec:opt} to solve the receive power or sum-rate maximization problems using the approximate channel models $\bH_{\text{app,2}}$ and $\bH_{\text{app,3}}$, i.e., $\bar{\bH}_{RT}$, $\bar{\bH}_{RI}$, $\bar{\bH}_{IT}$, and $\bar{\bthe}$  in problems \eqref{SISO:problem}, \eqref{MIMO:problem}, and \eqref{multiuser} are replaced by their approximate counterparts.

In Fig. \ref{MIMO_receive}, we depict the receive power for SISO and single-stream MIMO systems and the sum rate for multiuser MISO systems achieved by the three channel models. To better investigate the impact of mutual coupling among RIS antennas, we consider two different values of inter-element spacing, $d=\lambda/2$ and $d=\lambda/4$, respectively. 
As shown in the figure,  
the performance with and without the unilateral approximation is almost indistinguishable in all cases. In contrast,  mutual coupling among RIS antennas has a noticeable impact on the performance, particularly when the inter-element spacing is small. The presence of mutual coupling improves system performance compared with the case without coupling, and stronger coupling leads to more significant performance gains.
 We remark here that the benefits of mutual coupling under Rayleigh fading channels were previously reported for SISO systems in \cite{mutualcoupling2}. Our results demonstrate that the same conclusion also holds for MIMO and multiuser scenarios.  

In Fig. \ref{fig_relative}, we further evaluate the quality of the approximate solutions obtained using the channel models $\bH_{\text{app,2}}$ and $\bH_{\text{app,3}}$ to better visualize the accuracy of different approximations. The quality of each approximate solution is quantified by applying it to the general physics-consistent channel model $\bH$ and computing its relative performance compared with the solution obtained from $\bH$. From Fig. \ref{fig_relative}, we can draw similar conclusions  as those from Fig. \ref{MIMO_receive}. First, the unilateral approximation has a negligible effect in all cases, yielding nearly 100\% relative performance compared to the solution of the general physics-consistent model. Second, ignoring mutual coupling among RIS elements  leads to performance degradation, particularly when the inter-element spacing is small. For example, when $d=\lambda/4$, the solution from the mutual-coupling-unaware model achieves only around $65\%$ of the performance.

In Fig. \ref{fig_wrtK}, we investigate how the number of transmit/receive antennas and the number of users affect the accuracy of different approximations.  We consider a single-stream MIMO system and a multiuser MISO system, and  depict the relative performance of various approximate channel models as a function of the number of transmit/receive antennas and the number of transmit antennas/users, respectively. As can be observed, the unilateral approximation remains accurate across all considered systems, whereas the effect of mutual coupling becomes more pronounced as the number of transmit/receive antennas or users increases.

\begin{figure}
\centering
\subfigure[Single-stream MIMO system.]{\includegraphics[width=0.33\textwidth]{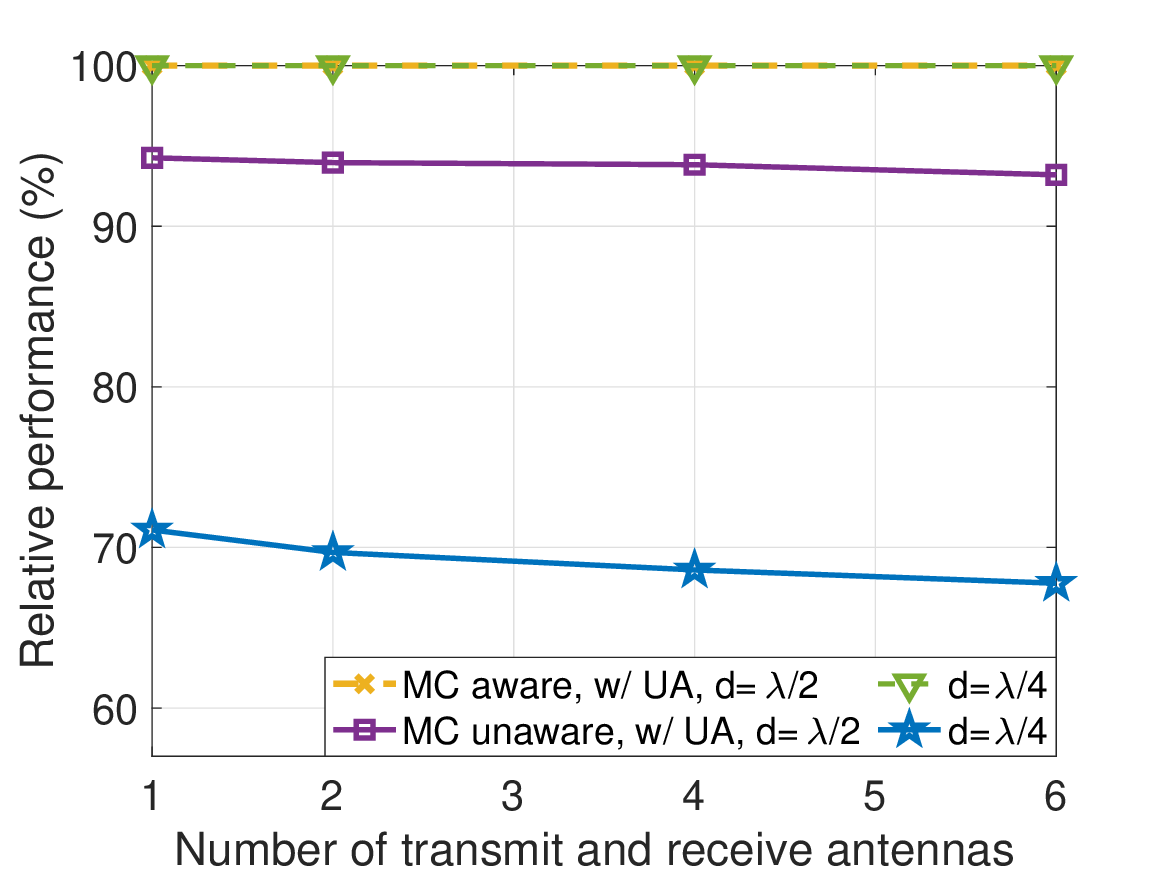}}
\subfigure[Multiuser MISO system.]{\includegraphics[width=0.33\textwidth]{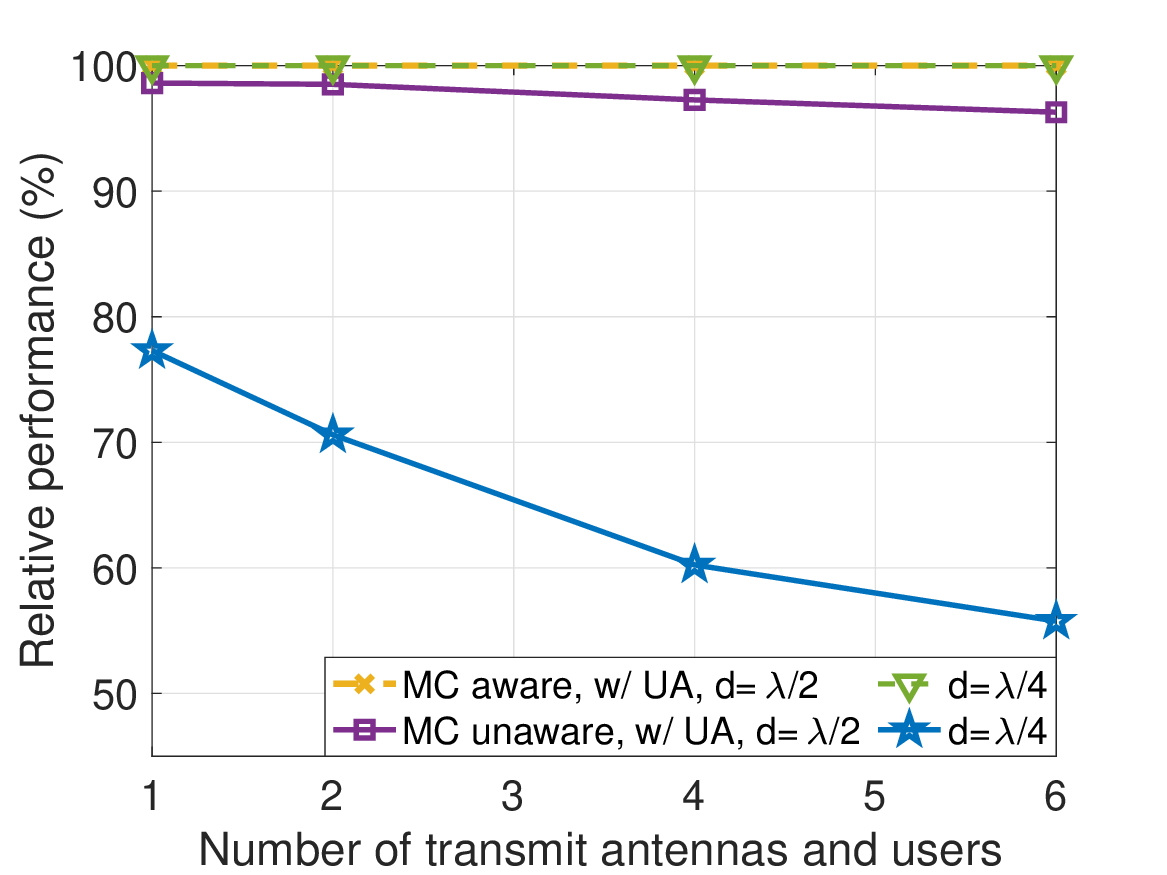}}
\caption{Relative performance of the solutions obtained from approximate channel models versus the number of transmit antennas and receive antennas/users, where $N_T=N_R$.  The number of RIS elements is fixed as $N_I=64$.}
\label{fig_wrtK}
\end{figure}
\begin{figure}
\centering
\subfigure[SISO and single-stream MIMO systems.]{\includegraphics[width=0.33\textwidth]{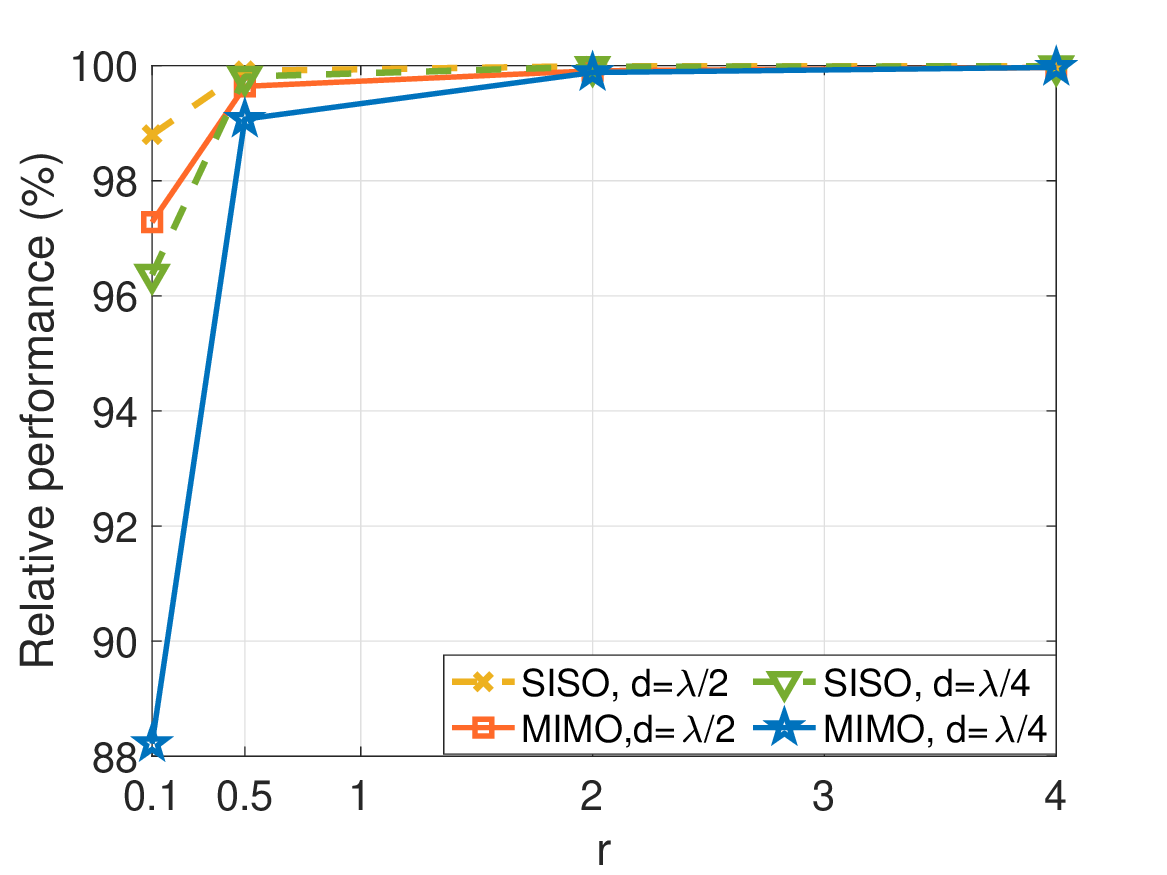}}
\subfigure[Multiuser MISO system.]{\includegraphics[width=0.33\textwidth]{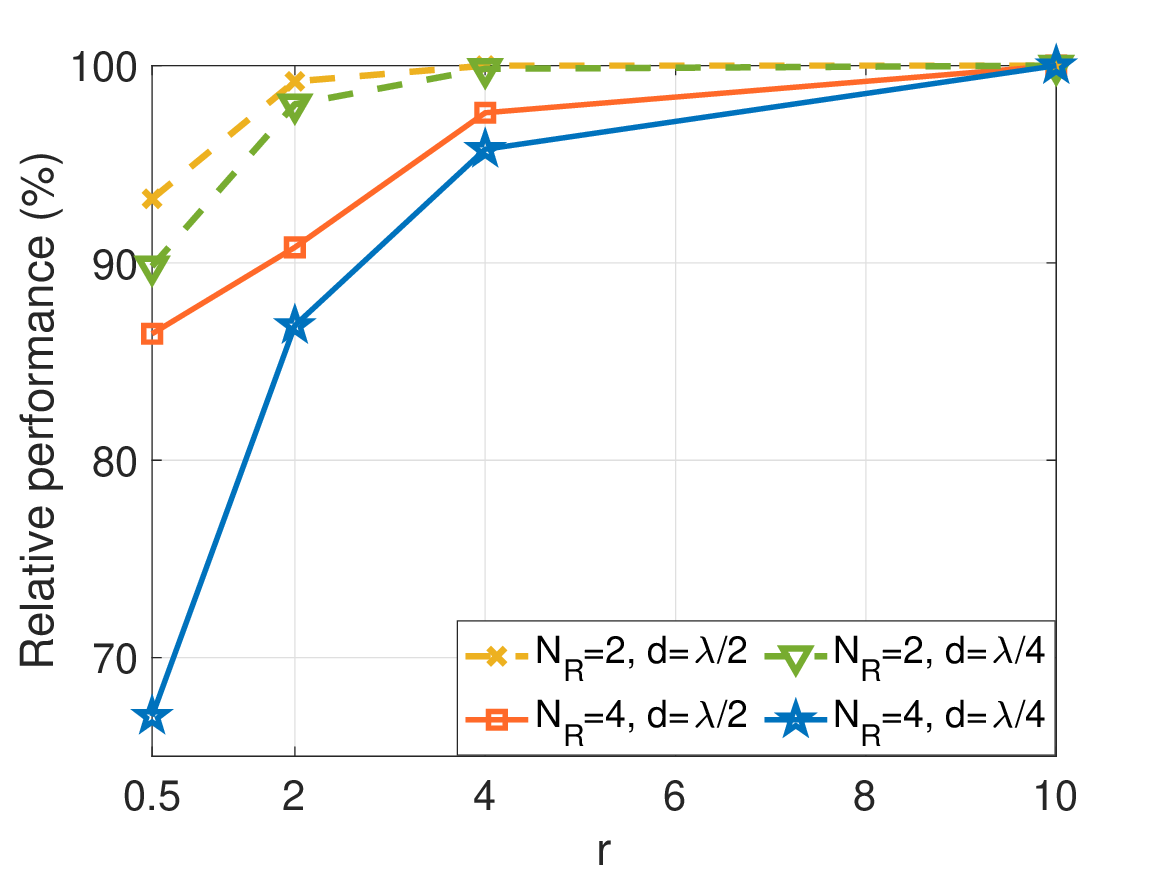}}
\caption{Relative performance of the solution obtained from approximate channel model $\bH_{\text{app,2}}$ versus 
the distance (normalized by wavelength) between transmitter and RIS. The number of RIS elements is fixed as $N_I=64$. The number of transmit and receive antennas for the MIMO system is $N_T=N_R=4$. The numbers of transmit antennas and users for the multiuser MISO system are set as the same. }
\label{fig_smalld}
\end{figure}

In all the above figures, the distances between  transmitter, RIS, and  receiver are set to be much larger than the wavelength, which is typically the case in practice. It is commonly accepted that the unilateral approximation accurately represents real world channels under this condition, as also validated by our results. In Fig. \ref{fig_smalld}, we further investigate the range of validity of the unilateral approximation by considering an extreme scenario where the distance between the transmitter and RIS is on the order of the wavelength. 
We assume that mutual coupling is aware and focus on the effect of unilateral approximation by comparing the performance achieved with models $\bH_{\text{app,2}}$ and $\bH$. The number of RIS elements is  fixed as $N_I=64$.

As shown in Fig. \ref{fig_smalld}, the unilateral approximation remains accurate across a broad range, even when the distance between the transmitter and RIS is on the order of the wavelength. For example, even at $0.1\lambda$, the solution obtained under the unilateral approximation still achieves over 95\% of the performance of the exact solution for the SISO system.   Although the accuracy decreases with the the number of transmit/receive antennas and users, it is still sufficiently high for  all  normal RIS-aided wireless communication systems.

We remark that we have also investigated the impact of unilateral approximation and mutual coupling under different RIS architectures (other than the fully-connected RIS). The results exhibit trends consistent with those in Figs. \ref{fig_relative} -- \ref{fig_smalld}. This is expected, since the unilateral approximation and mutual coupling are related to EM effects among/of the transmit, RIS, and receive antennas, and are irrelevant to the circuit topology of the reconfigurable impedance network.   For brevity and clarity, we omit these additional figures.

In summary, we can draw the following conclusions from the simulation results regarding the effects of mutual coupling and unilateral approximation.  First, stronger mutual coupling among RIS elements can enhance system performance for all SISO, MIMO, and multiuser MISO systems under Rayleigh fading channels. Second, ignoring the mutual coupling among RIS elements in the optimization process leads to performance degradation, particularly when the inter-element spacing of RIS is small.  Third, the unilateral approximation is highly accurate  even when the distance between the transmitter and RIS is only a  few wavelengths. Therefore, it can be reliably employed in normal RIS-aided wireless communication  systems to simplify the channel model without loss of accuracy.

\section{Conclusion}\label{sec:7}
In this paper, we studied the architecture design and optimization of RIS under the general physics-consistent model derived in \cite{generalmodel}, which captures practical non-linearities such as imperfect matching and mutual coupling. To this end, we derived an explicit expression of the model that makes the impact of the RIS explicitly visible. Under the common assumptions that the RIS is lossless and reciprocal, we further derived a compact reformulation of the model, thereby facilitating both optimization and analysis.
Based on this formulation, we make the following two key contributions. First, we prove that band-connected RIS is optimal under the general physics-consistent model, which unifies and generalizes existing results in \cite{tree,mutualcoupling2,graph}. Second, we develop optimization algorithms for different systems under the general physics-consistent model. This enables us to give a comprehensive evaluation of the impact of various effects and approximations. Our simulations demonstrate that the unilateral approximation widely adopted in the literature is highly accurate for normal RIS-aided wireless communication systems. In contrast, mutual coupling among RIS elements has a significant influence on system performance: stronger mutual coupling can enhance performance under Rayleigh fading channels, and
neglecting it during optimization leads to performance degradation.  
\appendices
\section{Proof of Lemma \ref{lem:pd}}\label{app:pd}
For notational simplicity, let $\bX:=\bY_{IR}$, $\bM:=\mathbf{Y}_{RR}+\mathbf{Y}_R$, $\bM_R:=\Re(\bM)$, and $\bM_I:=\Im(\bM)$. Hence, 
\begin{equation}\label{pd:eq1}
\begin{aligned}
&\Re(\mathbf{Y}_{II}-\mathbf{Y}_{IR}(\mathbf{Y}_{RR}+\mathbf{Y}_R)^{-1}\mathbf{Y}_{RI})\\
&\overset{(a)}{=}\Re(\mathbf{Y}_{II})-\Re(\mathbf{X}\bM^{-1}\mathbf{X}^T)\\
&\overset{(b)}{=}\Re(\bY_{II})-\Re(\bX)\Re(\bM^{-1})\Re(\bX)^T+\Im(\bX)\Re(\bM^{-1})\Im(\bX)^T\\
&~~~~+\Re(\bX)\Im(\bM^{-1})\Im(\bX)^T+\Im(\bX)\Im(\bM^{-1})\Re(\bX)^T,
\end{aligned}
\end{equation}
where (a) uses $\bY_{RI}=\bY_{IR}^T$ and (b) is obtained by expanding $\Re(\mathbf{X}\bM^{-1}\mathbf{X}^T)$. 
According to \cite[Proposition 1]{mutualcoupling2}, the real part of the admittance matrix of a reciprocal and lossy network (which is the case for a wireless channel) is positive definite with probability one, i.e., $\Re(\mathbf{Y})\succ \mathbf{0}$. Hence, 
$$\left[\begin{matrix} 
\Re(\bY_{II})&\Re(\bY_{IR})\\\Re(\bY_{RI})&\Re(\bY_{RR})
\end{matrix}\right]\succ \mathbf{0}.$$
This further implies that $\Re(\bY_{RR})$ and its Schur complement  are positive definite \cite{boyd2004convex}, i.e., $\Re(\bY_{RR})\succ \mathbf{0}$ and 
\begin{equation}\label{Schur}
\Re(\mathbf{Y}_{II})-\Re(\mathbf{X})\Re(\mathbf{Y}_{RR})^{-1}\Re(\mathbf{X})^T\succ\mathbf{0},
\end{equation}
where we recall that $\bX=\bY_{IR}$.
The matrix $\bY_R$ is a diagonal matrix with diagonal elements denoting the load admittances, which are real and positive. Hence, 
 \begin{equation}\label{MR}
 \bM_R=\Re(\bY_{RR})+\bY_R\succ\Re(\bY_{RR}).
 \end{equation} This, together with \eqref{Schur}, further gives
\begin{equation}\label{pd:eq2}
\Re(\mathbf{Y}_{II})-\Re(\mathbf{X})\bM_R^{-1}\Re(\mathbf{X})^T\succ\mathbf{0}.
\end{equation}
By reorganizing the last line of \eqref{pd:eq1}, we get 
$$
\begin{aligned}
&\Re(\mathbf{Y}_{II}-\mathbf{Y}_{IR}(\mathbf{Y}_{RR}+\mathbf{Y}_R)^{-1}\mathbf{Y}_{RI})\\
&=\Re(\bY_{II})-\Re(\bX)\bM_R^{-1}\Re(\bX)^T\\
&~~~+\Re(\bX)(\bM_R^{-1}\hspace{-0.05cm}-\hspace{-0.05cm}\Re(\bM^{-1}))\Re(\bX)^T\hspace{-0.08cm}+\hspace{-0.05cm}\Im(\bX)\Re(\bM^{-1})\Im(\bX)^T\\
&~~~+\Re(\bX)\Im(\bM^{-1})\Im(\bX)^T+\Im(\bX)\Im(\bM^{-1})\Re(\bX)^T.
\end{aligned}
$$
According to \eqref{pd:eq2}, the above matrix is positive definite if
$$
\begin{aligned}
[\begin{matrix}\Re(\bX)&\Im(\bX)\end{matrix}]\left[\begin{matrix}
\bM_R^{-1}-\Re(\bM^{-1})&\hspace{-0.1cm}\Im(\bM^{-1})\\\Im(\bM^{-1})&\hspace{-0.1cm}\Re(\bM^{-1})\end{matrix}\right]\left[\begin{matrix}\Re(\bX)^T\\\Im(\bX)^T\end{matrix}\right]\succeq\mathbf{0}.
\end{aligned}$$
To establish this,   it suffices to show that 
\begin{equation}\label{pd:eq3}
\left[\begin{matrix}
\bM_R^{-1}-\Re(\bM^{-1})&\Im(\bM^{-1})\\\Im(\bM^{-1})&\Re(\bM^{-1})\end{matrix}\right]\succeq \mathbf{0}.
\end{equation}
Next, we prove \eqref{pd:eq3} by showing that $\Re(\bM^{-1})\succ \mathbf{0}$ and its Schur complement 
$$\bM_R^{-1}-\Re(\bM^{-1})-\Im(\bM^{-1})\Re(\bM^{-1})^{-1}\Im(\bM^{-1})\succeq\mathbf{0}.$$
Note that for a matrix $\bM$ in complex space, the real and imaginary parts of its inversion can be expressed as 
\begin{equation}\label{invM}
\begin{aligned}
\Re(\bM^{-1})&=(\bM_R+\bM_I\bM_R^{-1}\bM_I)^{-1},\\
\Im(\bM^{-1})&=-\bM_R^{-1}\bM_I(\bM_R+\bM_I\bM_R^{-1}\bM_I)^{-1}.
\end{aligned}
\end{equation}
  By \eqref{MR}, $\bM_R\succ \mathbf{0}$, and thus 
$\bM_R+\bM_I\bM_R^{-1}\bM_I\succeq \bM_R\succ\mathbf{0}$.  It follows that $\Re(\bM^{-1})\succ\mathbf{0}$. 
In addition, applying \eqref{invM} gives
$$
\begin{aligned}
&\bM_R^{-1}-\Re(\bM^{-1})-\Im(\bM^{-1})\Re(\bM^{-1})^{-1}\Im(\bM^{-1})=\mathbf{0}.
\end{aligned}$$
This completes the proof.
\section{Proof of Proposition \ref{pro:arch}}\label{app:proarch}
 Our proof relies on the following auxiliary result, which is implicitly proved in \cite[Section IV-B2]{graph}. 
\begin{lemma}[{\cite[Section IV-B2]{graph}}]\label{lemma:solution}
Given $\bX\in\R^{m\times n}$ and $\bY\in\R^{m\times n}$, where $m\geq n$. Consider the following system:
$$
\left\{
\begin{aligned}
&\bB\bX=\bY,\\
&\bB=\bB^T,~
[\bB]_{i,j}=0,~~\forall~|j-i|>n-1.
\end{aligned}
\right.
$$
The above system admits a solution $\bB\in\R^{n\times n}$ if (i) $\bX^T\bY=\bY^T\bX$, and (ii) $\bX$ has no singular submatrices. 
\end{lemma}
Without loss of generality, we focus on  $q=N_R$ and denote $\mathcal{H}_{\text{band}}:=\mathcal{H}_{\text{band},N_R}$. 
Let $\bU=(\bar{\mathbf{H}}_{RI}\bar{\bthe})^H$,  then the set $\mathcal{H}_{i}$, $i\in\{\text{fully},\text{band}\}$,  can be expressed as 
$
\mathcal{H}_{i}= \left\{\bar{\mathbf{H}}_{RT}+\bU^H\bar{\mathbf{H}}_{IT}\mid\bU\in\mathcal{U}_{i}\right\},
$
where 
$$
\begin{aligned}
\mathcal{U}_{\text{fully}}&=\left\{\bU=(\bar{\mathbf{H}}_{RI}\bar{\bthe})^H\mid\eqref{bartheta} ,~\bar{\bB}_I=\bar{\bB}_I^T,~ \bar{\bB}_I\in\R^{N_I\times N_I}\right\}
\end{aligned}$$
and
$$
\begin{aligned}
\mathcal{U}_{\text{band}}&=\left\{\bU=(\bar{\mathbf{H}}_{RI}\bar{\bthe})^H\mid \eqref{bartheta} \text{ and }\eqref{barBI},~\right.\\
&\left.~~~~~~~~~\bB_I=\bB_I^T,~ [\bB_I]_{i,j}=0\text{ if }|j-i|>2N_R-1\right\},
\end{aligned}$$
respectively.  Define
 $$
\begin{aligned}
\mathcal{U}:=\{\bU\mid \bU^H\bU=\bar{\mathbf{H}}_{RI}\bar{\mathbf{H}}_{RI}^H, ~\bU^T\bar{\mathbf{H}}_{RI}^H=(\bU^T\bar{\mathbf{H}}_{RI}^H)^T\}.
\end{aligned}$$
Following a similar procedure as in  \cite[Section IV-B2]{graph}, we have $\mathcal{U}_{\text{fully}}\subseteq\mathcal{U}$.  Next, we show that 
$\mathcal{U}\subseteq \mathcal{U}_{\text{band}}\cup\mathcal{N}_{\mathcal{U}},$
where $\mathcal{N}_\mathcal{U}$ is defined in \eqref{NU}. The desired result in Proposition \ref{pro:arch} then follows immediately   with 
\begin{equation}\label{N}
\mathcal{N}=\left\{\bar{\mathbf{H}}_{RT}+\bU^H\bar{\mathbf{H}}_{IT}\mid\bU\in\mathcal{U}_{\text{fully}}\cap\mathcal{N}_{\mathcal{U}}\right\}.
\end{equation}

Given  $\bU \in \mathcal{U}$, $\bU \in \mathcal{U}_{\text{band}}$ if and only if the following system admits a solution:
\begin{equation}\label{proof:pro2_2}
\left\{
\begin{aligned}
&{\mathbf{B}}_I\bar{\bM}_\bU=\bar{\boldsymbol{\Gamma}}_\bU,\\
 &\bB_I=\bB_I^T,~~ [\bB_I]_{i,j}=0\text{ if }|j-i|>2N_R-1,
\end{aligned}\right.
\end{equation}
where $\bar{\bM}_\bU=\mathcal{R}(\bar{\bY}_{II})^{-\frac{1}{2}}\bM_\bU$ and $$\bar{\boldsymbol{\Gamma}}_\bU= \Re(\bar{\bY}_{II})^{\frac{1}{2}}\boldsymbol{\Gamma}_\bU-\I(\bar{\bY}_{II})\Re(\bar{\bY}_{II})^{-\frac{1}{2}}\bM_\bU$$
with
$
\bM_\bU=\left[\Re(\mathrm{i}\bU+\mathrm{i}\bar{\mathbf{H}}_{RI}^H),~\I(\mathrm{i}\bU+\mathrm{i}\bar{\mathbf{H}}_{RI}^H)\right]$ and 
$
\boldsymbol{\Gamma}_\bU=[\Re(\bU-\bar{\mathbf{H}}_{RI}^H),\Im(\bU-\bar{\mathbf{H}}_{RI}^H)].
$ 
It is easy to check that $\bar{\bM}_\bU^T\bar{\boldsymbol{\Gamma}}_\bU$ is symmetric since  $\bM_\bU^T\boldsymbol{\Gamma}_\bU$ is symmetric  \cite[Lemma 5]{graph}.
Applying Lemma \ref{lemma:solution} with $\bX=\bar{\bM}_\bU\in\R^{N_I\times 2N_R}$ and $\bY=\bar{\boldsymbol{\Gamma}}_{\bU}\in\R^{N_I\times 2N_R}$, we can conclude that \eqref{proof:pro2_2} has a solution as long as $\bU\notin\mathcal{N}_\mathcal{U}$, where 
\begin{equation}\label{NU}
\begin{aligned}
\mathcal{N}_\mathcal{U}=\left\{\bU\mid\bar{\bM}_{\bU} \text{ has a singular submatrix}\right\}.
\end{aligned}
\end{equation}
Finally, $\mathcal{N}$ is a low-dimensional subspace of $\mathcal{H}_{\text{fully}}$ since $\mathcal{N}_\mathcal{U}$ is a low-dimensional subspace of $\C^{N_I\times N_R}$\cite[Theorem 2]{graph}.

  \bibliographystyle{IEEEtran}
\bibliography{IEEEabrv,BDRIS}

\end{document}